\documentclass[10pt,final]{IEEEtran}
\usepackage{psfrag}
\usepackage[final,dvips]{graphicx}
\usepackage{amsmath}
\usepackage{amssymb}
\usepackage{epsfig}
\usepackage{amsthm}
\usepackage[square,numbers,sort&compress]{natbib}
\usepackage[table,xcdraw]{xcolor}
\usepackage{array}
\usepackage{algorithm}
\usepackage[noend]{algpseudocode}
\usepackage{mathtools}

\DeclarePairedDelimiter\floor{\lfloor}{\rfloor}

\usepackage[labelformat=simple]{subcaption}

\usepackage{booktabs}
\DeclareCaptionLabelSeparator{periodspace}{.\quad}
\newcommand\equationspace{\vspace{-1.5ex}}

\DeclareGraphicsExtensions{pst,eps}
\newcommand{\h}{\mathbf{h}} 
\newcommand{\e}{\mathbf{e}}

\newcommand{\Hyp}{\mathcal{H}}
\newcommand{\A}{\boldsymbol{\Lambda}} 
 
\newcommand{\mean}{\boldsymbol{\mu}} 
\newcommand{\Cov}{\boldsymbol{\Sigma}}

\newcommand{\CN}{\mathcal{C}\mathcal{N}}
\newcommand{\ID}{\text{ID}}
\newcommand{\REQ}{\text{REQ}}
\newcommand{\CSI}{\text{CSI}}

\newcommand{\IMGMT}{\mathcal{I}_{\text{MGMT}}}
\newcommand{\IDTP}{\mathcal{I}_{\text{DTP}}}
\newcommand{\NRx}{N_{\text{Rx}}}
\newcommand{\pFA}{p_{\text{FA}}}
\newcommand{\pMD}{p_{\text{MD}}}
\newcommand{\pAttack}{p_{\text{Attack}}}

\newcommand{\KRice}{K_{\text{Rice}}}
\newcommand{\KRiceE}{K_{\text{Rice},E}}
\newcommand{\Kd}{K_{\text{d}}}
\newcommand{\Kactive}{K_{\text{Active}}}
\newcommand{\Nframe}{N_{\text{Frame}}}

\newcommand{\prob}{\mathbb{P}}
\newcommand{\E}{\mathbb{E}}
\newcommand{\MelA}{\mathcal{M}_{\mathcal{A}}}
\newcommand{\Ker}{\mathcal{K}}
\newcommand{\KRC}{K_{\text{RC}}}
\newcommand{\KCN}{K_{\text{CN}}}
\newcommand{\MelS}{\mathcal{M}_{\mathcal{S}}}

\newcommand{\Dactive}{D_{\text{Active}}}
\newcommand{\Dsybil}{D_{\text{Sybil}}}
\newcommand{\Ksybil}{K_{\text{Sybil}}}

\DeclareMathOperator{\argmin}{arg\,min}
\theoremstyle{plain}
\newtheorem{thm}{Theorem}
\newtheorem{lemma}{Lemma}
\newtheorem{remark}{Remark}

\newtheorem{cor}{Corollary}

\newcommand{\numberthis}{\addtocounter{equation}{1}\tag{\theequation}}
\DeclareMathOperator{\Tr}{tr}

\title{Physical Layer Authentication in Mission-Critical MTC Networks: A Security and Delay Performance Analysis}
\author{
	\IEEEauthorblockN{Henrik Forssell, Ragnar Thobaben, Hussein Al-Zubaidy, James Gross}
   	\\\IEEEauthorblockA{Department of Information Science and Engineering, KTH Royal Institute of Technology, Stockholm, Sweden}
	\\\{hefo, ragnart, hzubaidy, jamesgr\}@kth.se
	\thanks{This work is supported in part by the Swedish Civil Contingencies Agency, MSB, through the CERCES project.}
}

\begin{document}
\maketitle

\begin{abstract}
We study the detection and delay performance impacts of a feature-based physical layer authentication (PLA) protocol in mission-critical machine-type communication (MTC) networks. The PLA protocol uses generalized likelihood-ratio testing based on the line-of-sight (LOS), single-input multiple-output channel-state information in order to mitigate impersonation attempts from an adversary node. We study the detection performance, develop a queueing model that captures the delay impacts of erroneous decisions in the PLA (i.e., the false alarms and missed detections), and model three different adversary strategies: data injection, disassociation, and Sybil attacks. Our main contribution is the derivation of analytical delay performance bounds that allow us to quantify the delay introduced by PLA that potentially can degrade the performance in mission-critical MTC networks. For the delay analysis, we utilize tools from stochastic network calculus. Our results show that with a sufficient number of receive antennas (approx. 4-8) and sufficiently strong LOS components from legitimate devices, PLA is a viable option for securing mission-critical MTC systems, despite the low latency requirements associated to corresponding use cases. Furthermore, we find that PLA can be very effective in detecting the considered attacks, and in particular, it can significantly reduce the delay impacts of disassociation and Sybil attacks.
\end{abstract}


\vspace{-2ex}
\section{Introduction}
\label{sec:introduction}
\IEEEPARstart{A}{s} mission-critical machine-type communication (MTC) emerges as a new approach to interconnect cyber-physical infrastructures, also new requirements on security features arise. Mission-critical machine-type communication targets at low latencies and high transmission reliabilities, in order to realize new use cases for instance arising in industrial automation. Thus, while in human-oriented communication data confidentiality followed by integrity form the utmost priorities (while service availability and security overhead typically have less relevance), the priorities change in the mission-critical setting. In detail, the order of concern is reversed~\cite{3gpp.22.804}: Service availability has highest priority since automation applications are typically supposed to run uninterrupted over long time spans. The second highest priority has message integrity, as in a closed control loop it is of vital importance that sensor and actuation information is not altered during transmission, while it also must be assured that the received data indeed stems from the claiming source. Finally, confidentiality is of least importance, as in automation applications the reading of sensor and actuation information poses only little threat to the controlled plant. Paired with the general requirement for low transmission latencies, these inverted security priorities are challenging, as traditionally integrity is assured through crypto schemes on the higher layers, which comes with significant computational complexities.

Physical layer authentication (PLA) has been proposed as a lightweight alternative for crypto security for authentication in reliable MTC communications~\cite{Weinand2017_2}. In general, PLA schemes perform hypothesis testing based on dedicated features of the communication pair like, e.g., the location-specific channel frequency response~\cite{Xiao2007} or a device-specific local oscillator offset~\cite{Hou2014} to determine if transmissions originate from legitimate sources. The advantage of this method is that messages can be authenticated quickly at the physical layer, without relying on cryptographic methods at higher layers and with slim-to-none security overhead. However, such schemes also come with drawbacks. First of all, due to the hypothesis testing PLA inevitably results in false alarms from time to time (i.e., some legitimate messages will be erroneously rejected) which can necessitate a retransmission. Furthermore, missed detections (i.e., accepting messages from an adversary) can occur if communication is subject to an impersonation attack. Thus, despite the complexity advantages, PLA also comes with costs which potentially can be significant in the context of mission-critical MTC.  
This raises the question how these costs (i.e., false positives and missed detections) potentially impact in particular the delay performance of a mission-critical MTC system.

Related work so far has largely not been addressing this question.
PLA for mission-critical MTC is proposed for instance in~\cite{Weinand2017,Weinand2017_2} but without considering the impact on the delay. 
In \cite{Wang2016}, the reduction in delay from removing authentication-induced processing delays in cell-handovers by using PLA is simulated.
However, this paper does not focus on MTC and additionally does not take false alarms of PLA into account.
Ozmen \textit{et al.} considers the delay-sensitive performance of a communication system under information-theoretic secrecy~\cite{Ozmen2017,Ozmen2018}. 
Delay in these works is characterized through the concept of effective capacity, which essentially allows the approximation of queuing-related performance metrics like the backlog or latency.
Furthermore in \cite{Naghibi2017}, the delay performance of a Rayleigh fading wiretap channel is studied using \textit{stochastic network calculus} for queueing analysis.
All papers~\cite{Ozmen2017,Ozmen2018,Naghibi2017} apply queueing analysis tools to study the delay impacts of different physical layer security techniques, however, none of them consider PLA.

In this work, we address the issue of delay analysis, and thus the cost, of PLA for mission-critical MTC.
We consider a centralized MTC network running a mission-critical application in which devices need to deliver data to the access point reliably and with low latency.
In the considered network, we introduce a standard generalized-likelihood-ratio test PLA scheme, which we extend to take multiple-message authentication into account. 
We model several strategies that the adversary can use, namely data injection, disassociation, and Sybil attacks, and analyze the detection performance for each scenario. 
To derive the delay performance impacts, we develop link-level queueing models that take the PLA errors and actions of the adversary into account. 
For queuing analysis, we employ tools from stochastic network calculus~\cite{Fidler2015,Schiessl2015}.
This work significantly extend the scope of our previous study \cite{Forssell2017} of delay impacts of PLA that only considered a single antenna system without an active attacker.

The contributions of our paper are the following: We derive delay performance bounds for MTC links where PLA is used for combating various attack strategies. We develop models for how data injection, disassociation, and Sybil attacks are launched against a MTC network and their impact on the queueing delay performance. With respect to stochastic network calculus, we provide an approximation to a previously unsolved mathematical problem: an upper bound on the delay violation probability over a Rice fading single-input multiple-output channel. From our results, we conclude that PLA, under relatively strong line-of-sight conditions and with sufficient number of receive antennas, can indeed provide high security in a mission-critical application. We also show that PLA effectively reduce the impact of disassociation and Sybil attacks at the cost of an approximately constant increase in delay violation probability. Thus, our results show that despite some costs, PLA promises to be an effective scheme in ensuring message integrity even in mission-critical MTC systems.

The rest of the paper is organized as follows: Section~\ref{sec:preliminaries} introduces the system assumptions and our problem formulation. In Section~\ref{sec:modeling}, we describe the attacker models and their impact on the queueing model. Section~\ref{sec:queueing_analysis} is devoted to deriving the delay performance bounds using tools from stochastic network calculus. In Section~\ref{sec:results}, we present our numerical results, and Section~\ref{sec:conclusion} concludes the paper.

\textit{Notation:} Matrices are represented by bold capital symbols $\mathbf{X}$, and $\mathbf{X}^T$ and $\mathbf{X}^\dag$ denote the matrix transpose and conjugate transpose, respectively. We let $\Tr(\mathbf{X})$ denote the trace of a matrix. Bold symbols $\mathbf{x}$ represents vectors with entries $x_i$ and $\mathbf{I}_N$ denotes the $(N\times N)$ identity matrix. We let $\|\mathbf{x}\|=\sqrt{|x_1|^2+...+|x_n|^2}$ be the Euclidian norm. For an event $E$, we let $\mathbb{P}(E)$ and $\mathbb{I}(E)$ denote the probability and indicator function, respectively. For a random variable $X$, $\mathbb{E}[X]$ denotes its expected value and $f_X(x)$ and $F_X(x)$ its probability density and cumulative distribution function, respectively. We let $\CN(\mean,\Cov)$ represent the multivariate complex Gaussian distribution with mean $\mean$ and covariance matrix $\Cov$, $\mathcal{N}(\mean,\Cov)$ the corresponding real-valued Gaussian distribution, $\chi^2_k$ a central $\chi^2$ distribution with $k$ degrees of freedom, and $\chi_k^2(\lambda)$ a non-central $\chi^2$ distribution with $k$ degrees of freedom and non-centrality parameter $\lambda$.

\vspace{-0.5ex}
\section{Preliminaries}
\label{sec:preliminaries}

\begin{figure*}[h]
\centering
\psfrag{A}[][]{\large Multiple-Antenna Access Point}
\psfrag{B}[][]{\scriptsize \textbf{MAC}}
\psfrag{C}[][]{\scriptsize Resource Scheduling}
\psfrag{D}[][]{\scriptsize \hspace{-8pt} \textbf{Physical Layer}}
\psfrag{E}[][]{\scriptsize Feature Bank}
\psfrag{F}[][]{\scriptsize Feature-Based Physical Layer Authentication}
\psfrag{G}[][]{\scriptsize Regular Physical Layer Processing}
\psfrag{H}[][]{\scriptsize Access Point}
\psfrag{I}[][]{\scriptsize Wireless MTC Device}
\psfrag{J}[][]{\scriptsize Adversary}
\includegraphics[width=1\textwidth]{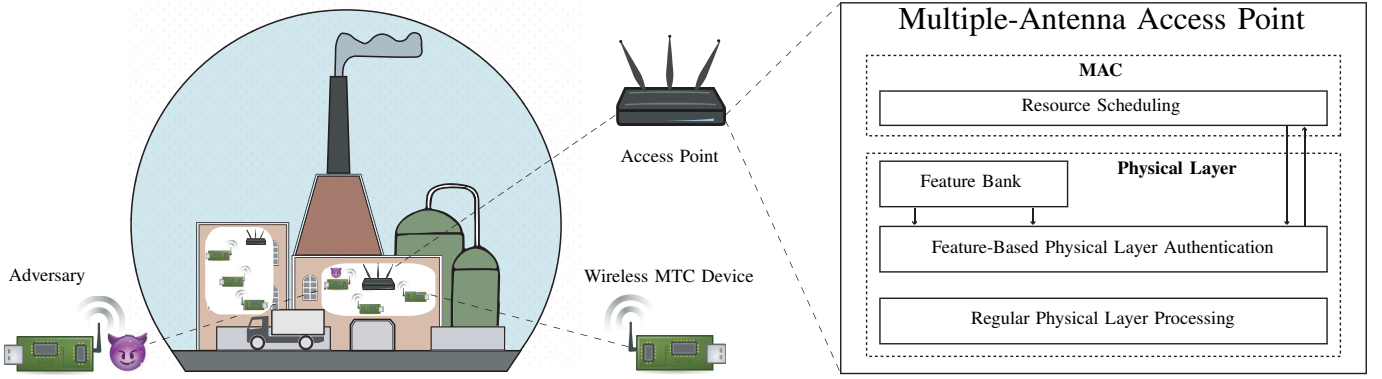}
 \vspace{-1ex}
\caption{Single-antenna MTC devices (e.g., wireless sensors in a critical monitoring application) communicating in uplink to a multiple-antenna access point. The access point is equipped with a feature-based PLA protocol.}
\label{fig:concept_picture}
\vspace{-3ex}
\end{figure*}

In this section, we present a centralized MTC network model consisting of $\Kd$ wireless devices communicating uplink data to an access point, as depicted in Fig.~\ref{fig:concept_picture}. The network is assumed to run a mission-critical application in which the MTC devices buffer data (e.g., sensor measurements) that need to be delivered reliably to the access point with minimal delay, as for example in motion control or generally in factory automation. Furthermore, as depicted in Fig.~\ref{fig:concept_picture}, we assume that there is an adversary present in the vicinity of the network, attempting to disturb the system using stealthy wireless impersonation attacks that are compliant with the typical behavior of legitimate devices within the network (e.g., sending payloads, data or disconnection requests). For protection against such attacks, the access point is using a feature-based physical layer authentication (PLA) protocol that compares the channel-state information associated with each transmission to a pre-stored feature bank. The access point is assumed to be equipped with $\NRx$ antennas, both in order to improve the PLA detection performance and to improve capacity, while the MTC devices (e.g., small sensors) are assumed to have single antennas. The stationary feature bank consists of the statistics of the phased-array antenna responses from each device to the multiple-antenna access point, and we assume that the devices are deployed such that a line-of-sight (LOS) path to the access point is available. 

\vspace{-0.5ex}
\subsection{Medium Access and Physical Layer}
\label{sec:phy_mac}
We assume that the MTC devices access the wireless medium in a frame-based structure, each beginning with a beacon transmitted by the access point for synchronization, followed by a management (MGMT) period where devices can make various requests\footnote{The MGMT phase is based on contention access; however, we assume that collisions are handled appropriately such that we can neglect their impact.}. A device can request connecting to the access point (CN), disconnecting (DCN), or resources for transmission of data payload (DTA). The allocation of resources is then communicated to the devices in a broadcast period (BP), followed by the data transmission period (DTP) where devices transmit buffered data. This medium-access model is similar to existing standards such as LTE~\cite{Mehaseb2016} and beacon-enabled IEEE 802.14e~\cite{Kurunathan2018}. We let $\IMGMT(k)$ denote the set of request messages received in the MGMT period in frame $k$, each associated with a request $\REQ(m)\in\{\text{DTA},\text{CN},\text{DCN}\}$ and a device identifier $\ID(m)\in\{1,\cdots,K_d\}$ (e.g., an identification code such as a MAC address). We denote by $\IDTP(k)$ the set of devices that are granted DTP resources in frame $k$ and we assume that the access point expects at most one request from each device.

We assume the DTP has a fixed length of $\Nframe$ complex symbols that are divided by TDMA to the devices in $\IDTP(k)$. A fair division of resources is assumed, where the number of symbols each device gets allocated in frame $k$ is denoted by\footnote{Note that $N_k$ in general is a random variable depending on the number of users allocated in the frame, and that $N_k $ can get very small if many devices request resources at the same time.}

\equationspace
\begin{equation}
N_k= \floor*{\frac{\Nframe}{|\IDTP(k)|}},
\label{eq:N_k}
\end{equation}
where $\floor{x}$ denotes the largest integer smaller than $x$. We denote the $(\NRx\times\Nframe)$ complex symbols received at the access point in frame $k$ by $\mathbf{Y}_k = [\mathbf{Y}_{k,i_1} \cdots \mathbf{Y}_{k,i_{|\IDTP|}}]$ and let $\mathbf{y}_{k,i}(n)$ denote the $n$th column of $\mathbf{Y}_{k,i}$ (i.e., the observation of the $n$th symbol received from device $i$ in frame $k$). The single-input multiple-output (SIMO) channel is modeled according to

\equationspace
\begin{equation}
\mathbf{y}_{k,i}(n) = \h_{k,i}x_{k,i}(n) + \mathbf{w}_{k,i}(n), \label{eq:simo}
\end{equation}
for $n\in \{1,\cdots,N_k\}$, where $\h_{k,i}$ represent the channel vector between device $i$ and the access point in frame $k$, $x_{k,i}(n)$ are the transmitted data symbols, and $\mathbf{w}_{k,i}(n)\sim \CN(\mathbf{0},N_0\mathbf{I}_{\NRx})$ is the additive noise represented by a circular symmetric complex Gaussian random vector. We assume that  $\mathbb{E}[\|\h_{k,i}\|^2] = P_i\NRx$ where $P_i$ represent the average power received per antenna from device $i$. We model the channel $\h_{k,i}$ as a narrowband SIMO Rice fading channel, i.e., $\h_{k,i}\sim \CN(\mean_i,\Cov_i)$ with $\mean_i$ representing the LOS component and covariance matrix $\Cov_i$ representing the fading. The covariance matrix is given by $\Cov_i = \frac{P_i}{K_{\text{Rice}}+1}\A$, where $[\A]_{i,j} = \rho^{|i-j|}$ is an $(\NRx \times \NRx)$ matrix, $\rho$ is a correlation coefficient, and $\KRice$ is a common Rice factor experienced by all antennas and all devices in the network. Furthermore, we assume that the frame period is shorter than the coherence time of the channel so that the channel realizations $\h_{k,i}$ can be assumed to be constant within a frame,  independent from frame to frame, and independent among the MTC devices.

For device $i$, positioned at distance $d_i$ and with angle of arrival (AoA) $\Phi_i$ relative to the receiver antenna array, the channel mean (i.e., the LOS component) is modeled as a phased-array antenna $\mean_i = a e^{-\frac{j2\pi d_i}{\lambda_c}}\e(\Omega_i)$, where $\lambda_c$ is the carrier wavelength, $\Omega_i = \cos(\Phi_i)$ is the directional cosine, $a=\|\mean_i\|$, and $\e(\Omega_i)$ is the unit spatial signature given by

\equationspace
\begin{equation}
\e(\Omega_i) = \frac{1}{\sqrt{\NRx}} \begin{bmatrix} z^0, z^{\Omega_i}, \cdots,z^{(\NRx-1)\Omega_i}\end{bmatrix}
\label{eq:norm_spatial_signature}
\end{equation}
in terms of the complex number $z=e^{-j2\pi\Delta_r}$, where $\Delta_r$ is the antenna spacing (normalized by the wavelength)~\cite{Tse2005}. From normalization of $\mathbb{E}[\|\h_{k,i}\|^2]$ we get $a = \sqrt{\frac{P_i \NRx K_{\text{Rice}}}{K_{\text{Rice}}+1}}$, and we assume the received power follows as $P_i = P_0 d_i^{-\beta/2}$ where $\beta$ is a path-loss exponent, $P_0$ is the transmit power, and $d_i$ is the distance. Additionally, in the following we normalize the noise power spectral density $N_0=1$ such that $P_i\NRx$ also represents the average received signal-to-noise ratio (SNR) on the $i$th link.

\begin{remark} The assumption of a narrow-band slow-fading LOS channel may appear as too restrictive at a first glance. However, it is relevant in scenarios with low or no device mobility and where the MTC deployment has been carefully planned, for instance, to use LOS beamforming for physical layer security~\cite{Abdelaziz2016}. Furthermore, these conditions allow us to upper bound the delay performance; if delay requirements are violated in this model, PLA will not be applicable for other models either. 
\end{remark}

\vspace{-0.5ex}
\subsection{Feature-Based Physical Layer Authentication}
\label{sec:feature_based_authentication}

With an adversary present, the validity of the message IDs are uncertain and the access point needs to determine their legitimacy. For PLA based on the observed channel states, the access point is assumed to have access to a feature bank consisting of the channel distributions $\CN(\mean_i,\Cov)$ that are associated with each legitimate channel and are used for hypothesis testing. In a real system, the access point can obtain the feature bank through learning based on legitimate transmissions (c.f., \cite{Xiao2016}). In this work, however, we assume that the distributions in the feature bank are perfectly known and the process by which they are obtained is omitted. For a received set of messages $\mathcal{I}=\{m_1,\cdots,m_M\}$, we denote by $\tilde{\h}_{m_i} = \CSI(m_i)$ the observed SIMO channel state associated with each message $m_i$. In general, this channel state is an estimate with limited precision. However, to simplify the analysis we assume perfect channel-state knowledge in the following. Furthermore, we assume that PLA is applied to $\mathcal{I} \in \{\IMGMT,\IDTP\}$, i.e., MGMT requests and DTP data payloads are authenticated separately. 

We consider now the case when $L$ messages share the same ID (e.g., due to multiple impersonated messages injected by an adversary). The PLA procedure divides the set $\mathcal{I}$ into subsets $\mathcal{I}_i = \{m\in\mathcal{I}: \ID(m) = i\}$ of messages with the same ID, each authenticated independently. To test the legitimacy of the messages in the set $\mathcal{I}_i$, the access point constructs a $L+1$-ary hypothesis test. We here denote by $\Hyp_l$ for $l\in\{1,\cdots,L\}$, the disjoint hypotheses that message $m_l$ is authentic, i.e., that we believe $\tilde{\h}_{m_l} \sim \CN(\mean_{\ID(m_l)},\Cov_{\ID(m_l)})$, and by $\Hyp_0$ the hypothesis that no message in $\mathcal{I}_L$ is authentic. The decision of $\Hyp_l$ results in accepting $m_l$ and rejecting the rest, while the decision of $\Hyp_0$ results in rejecting all messages in $\mathcal{I}_i$, since the authentication is predicated on that the access point expects only one message per legitimate device. The access point decides between the $L$ messages through

\equationspace
\begin{equation}
	 d_i(\tilde{\h}_{m_l}) \mathop{\gtrless}_{\mathcal{H}_l}^{\mathcal{H}_0} T, \quad\text{with } m_l = \underset{m=m_1,\cdots,m_L}{\argmin} d_i(\tilde{\h}_m),
	\label{eq:L_message_hyp_test_2}
\end{equation}
where $d_i(\cdot)$ is a discriminant function associated with the channel feature of the device with ID $i$, given by $d_i(\tilde{\h}_m) = 2(\tilde{\h}_m-\mean_i)^\dag\Cov^{-1}(\tilde{\h}_m-\mean_i)$. The minimization of the righthand side of \eqref{eq:L_message_hyp_test_2} is to be viewed as choosing the maximum-likelihood (ML) decision (the discriminant function $d_i(\cdot)$ is also the log-likelihood of the observation given the legitimate distribution) while the threshold decision in the lefthand side determines if the ML decision is authentic.

\paragraph*{Single message authentication ($L=1$)}
\label{par:sma}
The $L$ message hypothesis test in \eqref{eq:L_message_hyp_test_2} is an extension of the standard generalized-likelihood-ratio test (GLRT), used for PLA when deciding upon a multi-dimensional complex Gaussian feature such as a multi-carrier frequency response \cite{Xiao2008} or a channel impulse response~\cite{Mahmood2017}. Note that when \eqref{eq:L_message_hyp_test_2} is reduced to $L=1$ (i.e., only a single message with $\ID = i$ is received), the hypothesis test becomes $d_i(\tilde{\h}_m) \mathop{\gtrless}_{\mathcal{H}_1}^{\mathcal{H}_0} T$, where $\Hyp_1$ represents that the message is legitimate and $\Hyp_0$ represents that the message stems from an adversary.

\subsection{Adversarial Strategies}
\label{sub:adversary_assumptions_and_error_events}

In this paper, we assume that a single attacker is present in the system, referred to as Eve, having a single antenna, located at distance $d_E$ and with AoA $\Phi_E$ relative to the access point. We model Eve's channel similarly to the legitimate channels with Rice factor $\KRiceE$ and denote Eve's channel realization in frame $k$ by $\h_{k,E}~\sim \CN(\mean_E,\Cov_E)$, where $\mean_E = a_E e^{-\frac{j2\pi d_E}{\lambda_c}}\e(\Omega_E)$ with the normalized spatial signature given in \eqref{eq:norm_spatial_signature}. With this representation, we can model both the case when Eve is an external device or when the attack is launched from a compromised device within the network by letting $\mean_E=\mean_i$ and $\Cov_E=\Cov$ for some legitimate device $i$. The power received from Eve's transmissions is assumed to be $P_E=P_0d_E^{-\beta/2}$.

Given Eve's ability to send messages with fraudulent IDs, we differentiate four cases of adversary behavior:

\paragraph{Baseline} 
\label{par:baseline}
Eve is present, but inactive, and the performance of the system is only affected by false alarms. The baseline scenario models the impact of introducing the PLA protocol in the system when no attacks are attempted.

\paragraph{Data Injection Attack} 
\label{par:steady_state}
Eve is sending DTA requests impersonating a legitimate MTC device. Once successful, Eve gets DTP resources and transmits false data with the aim of harming the underlying application (e.g., drive a control system into a dangerous state by introducing fake sensor or actuation signals). In our work, we do not model the impact of the data injection attack on the application; however, metrics like missed detection rate (see Section~\ref{subsec:far_mdr} and~\ref{sec:data_injection}) measure Eve's success-rate under such attacks, and the number of resources $N_k$ each device gets scheduled will be affected.

\paragraph{Sybil Attack} 
\label{par:sybil_attack}

Eve transmits multiple DTA requests with fraudulent IDs, referred to as Sybil IDs/devices, with the goal of depleting resources available to the other legitimate devices~\cite{Xiao2009_Sybil}. In a Sybil attack, we assume that Eve targets a set of inactive devices $\Dsybil\subset \{1,\cdots,K_d\}$ that are not transmitting in the frame and sends DTA requests with the corresponding IDs. Note that it does not make sense for Eve to target active devices in this attack since they will already transmit DTA requests. With each successful Sybil ID, $N_k$ in \eqref{eq:N_k} is reduced which degrades the performance of the other links in the network.

\paragraph{Disassociation Attack} 
\label{par:disassociation_attack}
Eve targets a particular device and sends fraudulent requests to disassociate from the access point (DCN) with the corresponding device's ID. If successful, Eve disconnects the legitimate device which needs to reconnect, a process we model as being disconnected for $\KRC$ frames (e.g., due to management processes such as generating session keys).

The impersonation attacks that we consider can be launched by external entities (e.g., an attacker positioned in close proximity to the system, using a stolen MTC device or a software defined radio unit) or internal devices whose behavior has been hijacked by malicious code. Our attacker model allows us to model both cases by modifying the assumptions on Eve's channel. We note, however, that Sybil attacks are generally assumed to originate from internal devices that are compromised \cite{Xiao2009_Sybil}.

\vspace{-0.5ex}
\subsection{False Alarm and Missed Detection Rates}
\label{subsec:far_mdr}
Here, we summarize the error events and corresponding probabilities for the single message authentication, which are standard results (c.f., \cite{Xiao2008} for proofs). In the $L=1$ message case, two error events can occur: (i) a \textit{false alarm} when a legitimate message is rejected; and (ii) a \textit{missed detection} when an illegitimate message is accepted. Under the legitimate hypothesis $\Hyp_1$, we have $d_i(\tilde{\h}_m)\sim \chi^2_{2\NRx}$ and the false alarm rate is

\equationspace
\begin{equation}
\pFA(T) = \mathbb{P}(d_i(\tilde{\h}_m) >T|\mathcal{H}_1) = 1-F_{\chi^2_{2\NRx}}(T),
\label{eq:pFA}
\end{equation}
where $F_{\chi^2_{2\NRx}}(\cdot)$ is the cumulative distribution function (CDF) of a $\chi^2$ distribution with $2\NRx$ degrees of freedom. Observe that for a given choice of threshold $T$, the false alarm rate is equal across all device IDs $i$, independently of our assumptions on Eve. In practice, the PLA could be designed with different thresholds $T_i$ for different devices. However, in order to simplify the analysis we assume a constant threshold $T$. Under $\Hyp_0$ (i.e., Eve is sending the message $m$ with $\ID(m)=i$), given that Eve's channel covariance-matrix is of the form $\Cov_E=\frac{P_E}{1+\KRiceE}\A$, we have $d_i(\tilde{\h}_m)\sim \lambda_i \chi^2_{2\NRx}(\nu_i)$, where $\lambda_i = \frac{P_E(1+\KRice)}{P_i(1+\KRiceE)}$ and $\nu_i$ is the non-centrality parameter. Hence, the missed detection rate is

\equationspace
\begin{equation}
\pMD(i,T) = \mathbb{P}(d_i(\tilde{\h}_m)<T|\mathcal{H}_0) = F_{\chi^2_{2\NRx}(\nu_i)}(T/\lambda_i),
\label{eq:pMD}
\end{equation}
where $F_{\chi^2_{2\NRx}(\nu_i)}(\cdot)$ is the CDF of a non-central $\chi^2$ distribution with $2\NRx$ degrees of freedom and non-centrality parameter $\nu_i = 2(\mean_E-\mean_i)^\dag\Cov_E^{-1}(\mean_E-\mean_i)$. From this we can note that the missed detection rate varies with the device $i$ that Eve tries to impersonate. Error analysis for PLA with $L>1$ has to our knowledge not been studied before. In Section~\ref{sec:data_injection}, we provide bounds on the missed detection rate for $L=2$ and show that this case will suffice for the delay performance analysis under the considered attack strategies.

\vspace{-0.5ex}
\subsection{Delay Performance Metric}
\label{sec:delay_performance_metric}

As mentioned in Section~\ref{sec:introduction}, the use of PLA for improved security might have unintended consequences on the system's ability to meet delay requirements. To study such delay performance issues, we introduce infinite-buffer queues that model the flow of data from each MTC device to the access point. The queueing model is described by the bivariate stochastic processes
\begin{equation*}
A_i(\tau,t) = \sum_{k=\tau}^{t}a_k^{(i)},  \hspace{3pt}
D_i(\tau,t) = \sum_{k=\tau}^{t}d_k^{(i)}, 
\end{equation*}
representing the cumulative arrivals to and departures from the queue in the time interval $[\tau,t)$ for all $0 \le \tau \le t$. In frame $k$, $a_k^{(i)}$ represents the instantaneous arrivals to the $i$th MTC device buffer measured in bits (e.g., incoming sensor measurements), and  $d_k^{(i)}$ represent the instantaneous departures from the $i$th queue (i.e., information successfully transmitted to the access point). The ability to transfer data from the buffer queue to the destination at the access point is characterized by the cumulative service process $S_i(\tau,t) = \sum_{k=\tau}^{t}s_k^{(i)}$. Considering that a device is assigned resources, we assume that the transmitter chooses a coding rate $R_k^{(i)}$, and transmits $s_k^{(i)} = N_kR_k^{(i)}$ encoded information bits over the SIMO channel. Furthermore, we introduce the Bernoulli random variable $X_{k}^{(i)}$, indicating if resources are scheduled to device $i$. This results in the general service model

\equationspace
\begin{equation}
s_k^{(i)} = \begin{cases}
N_kR_k^{(i)}, & \text{if } X_{k}^{(i)}= 1  \\
0  & \text{if } X_{k}^{(i)}= 0.
\end{cases}
\label{eq:service_model}
\end{equation}
We use the Shannon capacity $R_k^{(i)} = \log_2(1+\gamma_{k,i})$ as a proxy for the amount of bits per channel use that can be transmitted over the channel. Assuming the access point has perfect channel state information and uses maximum-ratio combining for the channel model \eqref{eq:simo}, the instantaneous SNR is given by $\gamma_{k,i} = \frac{\|\h_{k,i}\|^2}{N_0}$.

A widely used measure on the queueing system's ability to meet delay requirements is the \textit{delay violation probability}~\cite{Zubaidy2016}. The queueing delay at time point $t$ is defined as

\equationspace
\begin{equation}
W_i(t) \triangleq \inf \{u>0;A_i(0,t)\leq D_i(0,t+u)\},
\end{equation}
representing the frames required to serve the bits in the queue at time $t$. This delay is randomly varying due to the random service process  and the delay violation probability is defined as $p_i(w) = \mathbb{P}(W_i(t)>w)$, i.e., the probability that a bit is not received within a defined deadline $w$. In many cases, an exact expression for the delay violation probability is complicated to derive. However, queueing analysis can give statistical bounds on this function. In particular, the stochastic network calculus framework, introduced in Section~\ref{sec:queueing_analysis}, contains tools that are appropriate for deriving an upper bound on $p_i(w)$ given the underlying service process in \eqref{eq:service_model}. Such delay bounds are particularly suitable for performance evaluation in mission-critical networks since they provide upper limits on the delay violation probability, i.e., a real system operating under the assumed conditions will certainly achieve a better delay performance.

\vspace{-0.5ex}
\subsection{Problem Formulation} 
\label{sec:problem_formulation}

Based on the system preliminaries outlined above, we are interested in jointly studying the security and delay performance impacts of PLA in the baseline scenario when Eve is inactive, as well as under the considered adversarial strategies presented in Section~\ref{sub:adversary_assumptions_and_error_events}. To be able to do this, we must first characterize how the PLA error events and the different attack strategies affect the link layer performance of the system, which we capture through queuing analysis. That is, we seek the distributions of $N_k$ and $X_k^{(i)}$ given the behavior of Eve. This problem is addressed in Section~\ref{sec:modeling}. Next, we must analyze how the PLA impacts the delay performance in the resulting queueing system. We tackle this by deriving upper bounds on the delay violation probability $p_i(w)$, subject to a given PLA threshold $T$, corresponding $\pFA(T)$ and $\pMD(i,T)$, and the adversary strategy. Derivations of the bounds are provided in Section~\ref{sec:queueing_analysis}. Based on this analysis, we seek to answer if, and under which circumstances, PLA is a viable option for authentication in mission-critical communications. More specifically, we want to answer what the baseline delay impacts on introducing PLA are, how detection and delay performance scale with the number of receive antennas $\NRx$ and the strength of LOS component $\KRice$, and what impacts the considered adversarial strategies have on the system. These among other questions are finally studied through our numerical results in Section~\ref{sec:results}.

\vspace{-2ex}
\section{Attack Modeling and Queueing Impacts}
\label{sec:modeling}
In this section, we analyze how erroneous PLA decisions impact the system and queueing service models that we have introduced in Section~\ref{sec:preliminaries} under each of the adversarial strategies.

\subsection{Baseline Scenario} 
\label{sub:steady_state}
In the baseline scenario, the adversary is inactive and the queueing model is affected only by dropped messages due to false alarms. We assume that a set $\Dactive\subseteq \{1,\cdots,K_d\}$ of devices are active and that each has a constant arrival rate $\alpha_i$, which means that each of the active devices will request DTA resources in each frame. Considering one of the active devices $i$, it will request resources with a DTA request in the MGMT period. Since the adversary is inactive, the access point will receive only one request with the ID of device $i$ and the message will be authenticated based on the single message authentication $d_i(\tilde{\h}_m) \mathop{\gtrless}_{\mathcal{H}_1}^{\mathcal{H}_0} T$ (see Section~\ref{par:sma}). The observed channel state $\tilde{\h}_m$ will in this case be the authentic channel $\CN(\mean_i,\Cov)$ and the false alarm rate is given by $\pFA(T)$ in \eqref{eq:pFA}. Since we assume perfect channel-state information and a frame period shorter than the coherence time of the channel, the observed channel state will remain constant during the frame. Hence, if the DTA request is accepted, so will the following data payload message in the DTP\footnote{This is a consequence of our previous assumptions. However, if the coherence time is shorter, or estimation errors are present, modeling of this as a two independent authentication decisions would be straightforward.}. Since the requests independently get rejected by PLA with $\pFA(T)$, the number of scheduled devices follows a binomial distribution

\equationspace
\begin{equation}
	p_{|\IDTP|}(k) = \binom{|\Dactive|}{k}(1-\pFA(T))^k\pFA(T)^{|\Dactive|-k},
	\label{eq:binom_steady_state}
\end{equation}
and the distribution of $N_k$ follows as $p_{N_k}(n) = p_{|\IDTP|}(\frac{\Nframe}{n})$. The threshold $T$ is ideally set such that $\pFA(T)$ is low, giving a possible approximation $|\IDTP| \approx |\Dactive|$. For a particular device $i$, the distribution of $X_k^{(i)}$ is given by

\equationspace
\begin{equation}
	\Pr(X_k^{(i)}=0)=\pFA(T).
	\label{eq:steady_state_fa}
\end{equation}
That is, in case of a false-alarm in frame $k$, the data buffer observes zero service.

\vspace{-0.5ex}
\subsection{Detection of Data Injection Attacks}
\label{sec:data_injection}
In a data injection attack, Eve transmits a DTA request in the MGMT period with the aim of getting DTP resources for transmitting a false data message. Either Eve impersonates an inactive device $i$ that is not requesting resources in the current frame, in which case the DTA requests undergoes single-message authentication and is undetected with probability $\pMD(i,T)$, or Eve impersonates an active device, in which case the message is authenticated by $L=2$ message authentication. In the latter case, denoting by $m_i$ and $m_E$ the messages from device $i$ and Eve, respectively, a missed detection occurs in the union of events $\left\{\underset{m=m_i,m_E}{\argmin} d_i(\tilde{\h}_m) =m_E\right\}$ and $\{d_i(\tilde{\h}_{m_E})<T\}$. In this case, the probability of missed detection, denoted by $\pMD^{L=2}(i,T)$, can be written as

\equationspace
\begin{equation}
\begin{aligned}
\pMD^{L=2}(i,T) &= \prob(d_i(\h_E)<d_i(\h_i),d_i(\h_E)<T) \\ &= \prob(d_i(\h_E)<T)\prob(d_i(\h_i)>T) \\ &+ \prob(d_i(\h_E)<d_i(\h_i)|d_i(\h_i)<T).
\label{eq:p_e}
\end{aligned}
\end{equation}
We now use the notation $d_i = d_i(\h_i)$ and $d_E = d_i(\h_E)$ to discuss the probability \eqref{eq:p_e}. The second line of \eqref{eq:p_e} is simply $\pFA(T)\pMD(i,T)$. However, for the second term $\prob(d_E<d_i|d_i<T)$ an exact expression can only be obtained in integral form. Instead, by noting that $\prob(d_i<d_E,d_E<T) \leq \prob(d_E<T) = \pMD(i,T)$, we can provide upper and lower bounds

\equationspace
\begin{equation}
\pFA(T)\pMD(i,T) \leq \pMD^{L=2}(i,T) \leq \pMD(i,T).
\label{eq:upper_lower_p_e}
\end{equation}
Additionally, we can observe that $\prob(d_i<d_E,d_E<T) \leq \prob(d_E<d_i)$ and provide an upper bound $\prob(d_E<d_i) \leq p_d(i)$ in the following lemma:

\begin{lemma}
\label{thm:dis_attack}
The probability $\prob(d_E<d_i)$ can be upper bounded by

\equationspace
\begin{equation}
p_d(i) = \min_{-\lambda/2<t<1/2}(1+2(\lambda-1)t-4\lambda t^2)^{-\NRx}e^{-\frac{\nu_i \lambda t}{1+2\lambda t}},
\label{eq:p_d}
\end{equation}
where $\lambda_i = \frac{P_E(1+\KRice)}{P_i(1+\KRiceE)}$, and $\nu_i = 2(\mean_E-\mean_i)^\dag\Cov_E^{-1}(\mean_E-\mean_i)$.
\end{lemma}

\begin{proof}
We rewrite $\prob(d_i(\h_E)<d_i(\h_i)) = \prob(d_i(\h_i)-d_i(\h_E)>0)$ and use the Chernoff bound to get that for every $t>0$

\equationspace
\begin{align*}
	\prob(d_i(\h_i)-d_i(\h_E)>0) & = \prob(e^{t(d_i(\h_i)-d_i(\h_E))}>1) \\
	\leq \E\left[e^{t(d_i(\h_i)-d_i(\h_E))}\right]
	& = \E\left[e^{t d_i(\h_i)}\right]\E\left[e^{-t d_i(\h_E)}\right],
	\numberthis
	\label{eq:chernoff}
\end{align*}
where we have applied the Markov inequality and used the independence of $\h_i$ and $\h_E$. Now since $d_i(\h_i)\sim \chi^2_{2\NRx}$ and $d_i(\h_E)\sim \lambda\chi^2_{2\NRx}(\nu_i)$ (see Section~\ref{subsec:far_mdr}), we get $\E[e^{td_i(\h_i)}] = (1-2t)^{-\NRx}$ and $\E[e^{-td_i(\h_E)}] = (1+2\lambda t)^{-\NRx}\exp\left(\frac{-\nu_i \lambda t}{1+2 \lambda t}\right)$ for $-\lambda/2<t<1/2$ from the standard moment generating functions for the corresponding distributions. Plugging these expressions into \eqref{eq:chernoff} and minimizing over $t$ yields \eqref{eq:p_d} which completes the proof.

\end{proof}

The upper bound that is tightest out of \eqref{eq:upper_lower_p_e} and \eqref{eq:p_d} depends on the authentication threshold $T$ (clearly $\pMD(i,T)\rightarrow 1$ as $T\rightarrow \infty$ and $\pMD(i,T)\rightarrow 0$ as $T\rightarrow 0$ while $p_d(i)$ is independent of $T$). Hence, we tighten our bound on the missed detection probability when Eve is launching a data injection attack against active device $i$ by

\equationspace
\begin{equation}
\pMD^{L=2}(i,T) \leq p_{\text{MD},\text{Upper}}(i) = \min\{\pMD(i,T),p_d(i)\}.
\label{eq:p_e_upper}
\end{equation}
This bound will additionally later prove useful when analyzing the disassociation attack in Section~\ref{sec:disassociation_attack}.

\begin{remark} In a data injection attack, the delay performance of legitimate devices will be affected since accepted DTA requests from Eve will reduce the amount of resources scheduled to other devices. However, this impact is principally the same as under the Sybil attack discussed in Section~\ref{sub:sybil_attack}. Therefore, we only use the data injection scenario to study the detection performance of PLA, leaving questions regarding queueing performance to be answered by the study of Sybil attacks.
\end{remark}

\vspace{-0.5ex}
\subsection{Queueing Impacts of Sybil Attacks} 
\label{sub:sybil_attack}

Recall that in a Sybil attack, Eve targets a set of inactive devices $\Dsybil$ and sends DTA requests with the corresponding IDs. Consequently, the access point receives messages from $\Dactive \cup \Dsybil$ in the MGMT period and needs to differentiate which ones are legitimate. If Eve successfully gets many DTA requests through, $N_k$ given by \eqref{eq:N_k} decreases and legitimate devices get less resources which can result in growing queue backlogs. Under a Sybil attack, we assume that active legitimate devices $\Dactive$ experience service dropouts modeled by $X_k^{(i)}$ the same way as in the baseline case \eqref{eq:steady_state_fa}; however, the distribution of $N_k$ is different due to Sybil IDs launched by Eve.

Assuming all IDs in $\Dactive \cup \Dsybil$ are distinct (Eve is assumed to target only inactive devices in the Sybil attack), the number of devices that get resources in the data transmission period is

\equationspace
\begin{equation}
	|\IDTP| = \sum_{i\in\Dactive \cup \Dsybil} \mathbb{I}(d_i(\tilde{\h}_{m_i})<T).
\end{equation}
We decompose $|\IDTP| = \Kactive + \Ksybil$ where $\Kactive$ is characterized by the baseline distribution of \eqref{eq:binom_steady_state} (i.e., requests rejected by false alarms) and

\equationspace
\begin{equation}
	\Ksybil = \sum_{i\in\Dsybil} \mathbb{I}(d_i(\h_E)<T)
\end{equation}
represents the number of Sybil IDs successfully launched by Eve. For a moderate number of Sybil IDs ($< 30$), the distribution of $\Ksybil$ can be combinatorially approximated as
\begin{equation}
\begin{aligned}
p_{\Ksybil}(k) \approx \sum_{B \in A_k}&\prod_{i \in B}\pMD(i,T)\times \\ & \prod_{j \in B^c}(1-\pMD(j,T)),
\label{eq:k_sybil_approx_pmf}
\end{aligned}
\end{equation}
where $A_k$ denotes the set of all size $k$ subsets of $\Dsybil$. This approximation stems from an assumption that the events $\{d_i(\h_E)<T\}_{i\in\Dsybil}$ can be approximated as independent, in which case $\Ksybil$ is Poisson-binomial distributed. Now the distribution of $|\IDTP|$ under a Sybil attack can be written as the convolution

\equationspace
\begin{equation}
p_{\IDTP,\text{Sybil}}(k) = \sum_{l=0}^{k} p_{\Kactive}(l)p_{\Ksybil}(l-k),
\end{equation}
from which the distribution of $N_k$ follows as $p_{N_k}(n) = p_{\IDTP,\text{Sybil}}(\frac{\Nframe}{n})$. 

The impact of the Sybil attack depends on the system's available resources $\Nframe$: If $\Nframe$ by design allows all devices to communicate simultaneously, the Sybil IDs will not have a substantial impact on the service of the legitimate devices. However, if the system is optimized to only have a subset of devices communicating at a time (e.g., in order to reduce latency or if only a subset of devices is involved in a particular sensing tasks), the result of launching multiple additional Sybil IDs might have severe impacts on the active legitimate devices. An alternative counter-strategy is to only accept requests from devices that are expected to transmit (e.g., sensors carrying relevant measurements for the running application). However, such application-layer information might not be available at the physical and MAC layers.

\vspace{-0.5ex}
\subsection{Queueing Impacts of Disassociation Attacks} 
\label{sec:disassociation_attack}

In a disassociation attack, Eve targets an active legitimate device and sends DCN request with the corresponding ID. In an attacked frame, the access point will observe two messages $m_1$ and $m_2$ with the same ID (i.e., $\ID(m_1) = \ID(m_2)$) and uses \eqref{eq:L_message_hyp_test_2} to decide which one is authentic. If the access point accepts the DCN request from Eve, the legitimate device will need to reconnect in order to continue its data transfer which results in a disruption of the communication (i.e., $s_k = 0$) for $\KRC$ consecutive frames which can lead to growing backlogs and increased delay. In principle, Eve could launch disassociation attacks against multiple links within the network. However, here we model the queueing impact when Eve targets a single device $i$.

In the disassociation attack, the frame-level service process $s_k$ in \eqref{eq:service_model} follows the same model as in the baseline scenario \eqref{eq:steady_state_fa} (i.e., frames are dropped with the false alarm rate and $N_k$ is given by its baseline distribution). We consider independent Bernoulli attack attempts from Eve with probability $\pAttack$ and to model the impact on the queueing performance, we divide the data flow from device $i$ to the access point into blocks consisting of $\KRC$ frames each and define the aggregated arrival process as $a'_l = \sum_{k=\KCN l}^{\KRC (l+1)-1} a_k$ and service process as

\equationspace
\begin{equation}
s'_l = \begin{cases}
\sum_{k=\KRC l}^{\KRC (l+1)-1} s_k, & \text{if } D_l = 0 \\
0  & \text{if } D_l = 1,
\end{cases}
\label{eq:dis_service_process}
\end{equation}
where $D_l$ is a Bernoulli random variable indicating a successful disassociation attack in the block. The distribution of $D_l$ is then given by

\equationspace
\begin{equation}
	\prob(D_l=1) = 1 - (1-\pMD^{L=2}(i,T) \pAttack)^{\KRC},
	\label{eq:p_d_l}
\end{equation}
where $\pMD^{L=2}(i,T)$ is the probability of accepting Eve's DCN message (i.e., the same situation as in the data injection attack and hence $\pMD^{L=2}(i,T)$ is given by \eqref{eq:p_e}). We recall from Section~\ref{sec:data_injection} that a closed form solution for $\pMD^{L=2}(i,T)$ is not available. However, since $\Pr(D_l=1)$ is monotonically increasing with $\pMD^{L=2}(i,T)\in[0,1)$, an upper bound on $\pMD^{L=2}(i,T)$ suffices to upper bound $\prob(D_l=1)$. An upper bound on $\pMD^{L=2}(i,T)$ is given by \eqref{eq:p_e_upper} and hence we get an upper bound $\prob(D_l=1) \leq 1 - (1-p_{\text{MD},\text{Upper}}(i)\pAttack)^{\KRC}$.

Our analysis of the disassociation attack serves as a worst-case model due to the upper bound on $\prob(D_l=1)$. However, since in the next Section~\ref{sec:queueing_analysis} we aim to upper bound the delay violation probability, an upper bound on the disassociation probability suffice for this purpose. Additionally, we acknowledge that other methods could be used for reducing the impact of disassociation attacks. For example, one could always choose DTA over DCN requests, which would render the disassociation attack harmless as long as an active device is targeted. However, we see this as an issue of protocol design and include disassociation attacks in our studies in the following.

\vspace{-1ex}
 \section{Delay Performance Analysis}
 \label{sec:queueing_analysis}
 
In this section, we derive delay performance bounds for the considered system using tools from stochastic network calculus. We begin by introducing necessary results and notation from the stochastic network calculus framework:

\vspace{-0.5ex}
\subsection{Stochastic Network Calculus}
\label{sec:snc}
Stochastic network calculus is a mathematical framework that allows us to analyze input-output relationships of stochastic queueing systems through, for example, performance bounds on delay or backlog given arrival and service distributions. For a complete overview of stochastic network calculus, we refer to \cite{Fidler2015}. The work in~\cite{Zubaidy2016} developed the stochastic network calculus framework for wireless fading links by observing that the analysis is simplified by converting the bivariate stochastic processes $A(\tau,t)$, $S(\tau,t)$ and $D(\tau,t)$ into $\mathcal{A}(\tau,t)\triangleq e^{A(\tau,t)}$, $\mathcal{S}(\tau,t)\triangleq e^{S(\tau,t)}$ and $\mathcal{D}(\tau,t)\triangleq e^{D(\tau,t)}$. This transformation allows the characterization of the random service process in terms of the varying instantaneous SNR due to fading. This is referred to as transforming the bit-domain processes into the SNR-domain since the processes become linear in the instantaneous SNR $\gamma_k$ instead of logarithmic. Arrival processes in the SNR-domain can then be seen as instantaneous SNR demands. In bit-domain, stochastic network calculus is based on a $(\min,+)$ dioid algebra over $\mathbb{R}^+$. Stochastic network calculus in the SNR-domain, on the other hand, is instead based on the $(\min,\times)$ dioid algebra since processes in the SNR-domain become multiplicative instead of additive. The performance bounds, which can be seen as variations of moment bounds, are derived in terms of Mellin transforms of the involved queueing processes. The Mellin transform of a random variable $X$, closely related to the moment-generating function (MGF), is defined as $\mathcal{M}_{X}(s) = \mathbb{E}[X^{s-1}]$.

The upper bound on the delay violation probability we utilize in this paper is given by Lemma~\ref{lem:kernel}:
\begin{lemma}
\label{lem:kernel}
For $s>0$,

\equationspace
\begin{equation}
	p(w) \leq \Ker(s,t+w,t),
	\label{eq:snc_delay_bound}
\end{equation}
where $\Ker(s,\tau,t)$ is called the kernel and given by

\equationspace
\begin{equation}
	\Ker(s,\tau,t) \triangleq \sum_{u=0}^{\min(\tau,t)}\MelA(1+s,u,t)\MelS(1-s,u,\tau),
	\label{eq:kernel}
\end{equation}
and $\MelS(s,\tau,t) = \E[\mathcal{S}(\tau,t)^{s-1}]$ and $\MelA(s,\tau,t) = \E[\mathcal{A}(\tau,t)^{s-1}]$ are Mellin transforms of the SNR-domain service and arrival processes. 
\end{lemma}
\begin{proof}
See Theorem 1 in \cite{Zubaidy2016}.
\end{proof}

With i.i.d. instantaneous arrivals and service, we can write $\MelS(s,\tau,t)=\MelS(s)^{t-\tau}$ and $\MelA(s,\tau,t)=\MelA(s)^{t-\tau}$, where $\MelS(s) \triangleq \E[e^{s_k(s-1)}]$ and $\MelA(s) \triangleq \E[e^{a_k(s-1)}]$ due to the independence of the instantaneous service and arrivals $s_k$ and $a_k$. Then, assuming stationarity of the underlying queueing processes, we let $t\rightarrow \infty$ in the righthand side of \eqref{eq:snc_delay_bound} and get

\equationspace
\begin{equation}
	\lim_{t\rightarrow \infty} \Ker(s,t+w,t) = \frac{\MelS(1-s)^w}{1-\MelA(1+s)\MelS(1-s)},
	\label{eq:steady_state_kernel}
\end{equation}
under the stability condition $\MelA(1+s)\MelS(1-s)<1$ required for the sum in \eqref{eq:kernel} to converge. Since Lemma~\ref{lem:kernel} holds for all $s>0$, it follows that minimization of \eqref{eq:steady_state_kernel} over $s>0$ gives us an asymptotic upper bound on the delay violation probability. Hence, for the stable and stationary queueing system, the upper bound on the delay violation probability can be compactly written as $p(w) \leq \inf_{s>0}\left\{\lim_{t\rightarrow \infty} \Ker(s,t+w,t) \right\}$
%
%
with the objective function to be minimized given by the steady-state kernel in \eqref{eq:steady_state_kernel}. This function can be shown to be a convex function for every $s$ in the stability interval $\MelA(1+s)\MelS(1-s)<1$ (see Theorem~1 in~\cite{Zubaidy2016_convex}). However, no analytical tools from convex optimization can be applied, and therefore, one typically resorts to a numerical grid search for the minimization over $s$.

Since in this paper we assume constant arrivals of $\alpha$ bits per frame, the arrival process is deterministic and the Mellin transform of the SNR-domain arrival process can easily be found to be $\MelA(s) = e^{\alpha(s-1)}$. The service process, following the SIMO channel service model \eqref{eq:service_model}, has a more complicated Mellin transform which we derive in the following subsections for the considered attack scenarios.

It is worth noting that alternative stochastic network calculus approaches exist that may be used for this analysis including \textit{effective capacity}~\cite{Wu2003} and \textit{MGF-based analysis}~\cite{Fidler2006}. Nevertheless, the usefulness of the approach in~\cite{Zubaidy2016_convex} that we employ is most apparent when applied to wireless fading channels as the Mellin transform $\mathcal M_{\mathcal S}$ is already derived for many fading channels in the literature, e.g.,~\cite{Schiessl2015,Zubaidy2017,Naghibi2017}. This makes the approach particularly attractive for wireless networks analysis.

\subsection{Baseline Analysis}
Recall that in the baseline scenario no active attacker is present and frames are dropped with the false alarm rate, i.e, $\Pr(X=0) = 1-p_X = \pFA$. The service model is given by \eqref{eq:service_model} with $R_k = \log_2(1+\h_k^\dag\h_k)$ where we now, for ease of notation, have dropped the user index $i$. Note that in this section we assume the allocated resources $N_k$ to be deterministic, something we will later generalize when deriving the Sybil attack bound. To simplify the derivation, we define the functions $h(\gamma_k,X_k) \triangleq e^{s_k}$ and $g(\gamma_k)=1+\gamma_k$ in terms of the instantaneous SNR $\gamma_k$ so that

\equationspace
\begin{align*}
h(\gamma_k,X_k) = \begin{cases}
g(\gamma_k)^{\frac{N_k}{\ln(2)}}, & \quad\text{if}\quad X_k=1 \\
1, & \quad\text{if}\quad X_k=0.
\end{cases}
\label{eq:h_function}
\numberthis
\end{align*}
In the following, we provide our main analytical result, which is an approximate expression for the Mellin transform of $g(\gamma_k)$ in Theorem~\ref{thm:mellin_simo}. From this result, the Mellin transform of the service process in steady-state easily follows, as stated in Corollary~\ref{cor:mellin_steady}.

\begin{thm}
\label{thm:mellin_simo}
For the Rice fading SIMO channel with mean $\mean$ and covariance matrix $\Cov$, the Mellin transform of $g(\gamma_k)$ can be approximated by

\equationspace
\begin{equation}
\begin{aligned}
	\mathcal{M}_{g(\Gamma_k)}(s) & \approx \frac{e^{1/2\alpha_g}(2\alpha_g)^{s-1}}{\Gamma(k_g/2)} \times \\ &\hspace{-1cm} \sum_{m=0}^{\infty}\binom{\frac{k_g-2}{2}}{m}\frac{1}{(-2\alpha_g)^m}\Gamma\left[s-m+\frac{k_g-2}{2},\frac{1}{2\alpha_g}\right]
	\label{eq:simo_thm}
\end{aligned}
\end{equation}
where $\Gamma(s,x)=\int_x^{\infty}t^{s-1} e^tdt$ denotes the upper incomplete gamma function and $\alpha_g$ and $k_g$ are parameters of the approximate distribution of $\gamma_k$ given by

\equationspace
\begin{equation}
\alpha_g = \frac{\frac{1}{2}(\Tr(\Cov^2)+\mean^\dag\Cov\mean)}{1+\Tr(\Cov)+\mean^\dag\mean} \quad \text{and} \quad k_g = \frac{(1+ \Tr(\Cov)+\mean^\dag\mean)^2}{\frac{1}{2}(\Tr(\Cov^2)+\mean^\dag\Cov\mean)}.
\end{equation}
\end{thm}

\begin{proof}

We begin by using the fact that $\gamma_k$ is a sum of independent non-central $\chi^2$ distributed random variables with $\mathbb{E}[\gamma_k] = \Tr(\Cov)+\mean^\dag\mean$ and $\text{Var}[\gamma_k] = \Tr(\Cov^2)+\mean^\dag\Cov\mean$. Now we use that the sum of non-central $\chi^2$ random variables can be approximated as a scaled central $\chi^2$~\cite{Pearson1959}. That is, we write $\gamma_k\approx\alpha_g X$, where $X\sim\chi^2_{k_g}$. Transferred to the Mellin transform, the approximation becomes $\mathcal{M}_{g(\h_k)}(s) \approx \mathcal{M}_{1+\alpha_g X}(s)$. Now, we seek

\equationspace
\begin{align*}
&\mathcal{M}_{1+\alpha X}(s)  = \int_{0}^{\infty}(1+\alpha x)^{s-1}\frac{1}{2^{\frac{k_g}{2}}\Gamma(k_g/2)}x^{\frac{k_g}{2}-1}e^{-\frac{x}{2}}dx\\
& = \frac{e^{\frac{1}{2\alpha}}}{(2\alpha)^{\frac{k_g}{2}}\Gamma(k_g/2)}\underbrace{\int_{1}^{\infty}u^{s-1}(u-1)^{\frac{k_g}{2}-1}e^{-\frac{u}{2\alpha}}du}_{\triangleq I},
\numberthis
\label{eq:simo_thm_1}
\end{align*}
where in the second line we have used the change of variable $u = 1+\alpha x$ and defined the integral $I$ which remains to be solved. To solve it, we can use the binomial expansion $(u-1)^{\frac{k_g}{2}-1} = \sum_{m=0}^{\infty}\binom{\frac{k_g-2}{2}}{m}(-1)^m u^{\frac{k_g}{2}-1-m}$, which plugged into the integral $I$ yields

\equationspace
\begin{align*}
&I
 = \int_{1}^{\infty}u^{s-1}\sum_{m=0}^{\infty}\binom{\frac{k_g-2}{2}}{m}(-1)^m u^{\frac{k_g}{2}-1-m}e^{-\frac{u}{2\alpha_g}}du \\
& \hspace{-0.6ex}= \sum_{m=0}^{\infty}\binom{\frac{k_g-2}{2}}{m}\frac{(2\alpha_g)^{s-m + \frac{k_g-2}{2}}}{(-1)^m}\int_{1/2\alpha_g}^{\infty}t^{s-m+\frac{k_g-2}{2}-1}e^{-t}dt \\
& \hspace{-0.6ex}= (2\alpha_g)^{k'+s}\sum_{m=0}^{\infty}\binom{k'}{m}\frac{1}{(-2\alpha_g)^m}\Gamma\left[s-m+k',\frac{1}{2\alpha_g}\right],\numberthis
\label{eq:simo_thm_2}
\end{align*}
where we in the second to third line have used the change of variable $t= u/2\alpha_g$ and introduced $k' = \frac{k_g-2}{2}$. Plugging \eqref{eq:simo_thm_2} into \eqref{eq:simo_thm_1} yields \eqref{eq:simo_thm}. Finally, since $\mathbb{E}[\alpha_g X] = \alpha_gk_g$ and $\text{Var}[\alpha_g X] = \alpha_g^2 2k_g$, we need $\alpha_g = \frac{\text{Var}[\gamma_k]}{2\mathbb{E}[\gamma_k]}$ and $k_g = \frac{2\mathbb{E}[\gamma_k]^2}{\text{Var}[\gamma_k]}$ in order to match the two first moments of the approximation, which completes the proof.

\end{proof}

With the result of Theorem~\ref{thm:mellin_simo} in place, we get the service-process Mellin transform through Corollary~\ref{cor:mellin_steady}:

\equationspace
\begin{cor}
	\label{cor:mellin_steady}
	For the baseline scenario, with Bernoulli frame drops with probability $\pFA$ due to PLA, the Mellin transform of the service process is given by
	\begin{equation}
	\begin{aligned}
		&{\MelS}_{,\text{Baseline}}(s) =\\& (1-\pFA)\mathcal{M}_{g(\gamma_k)}\left[1+\frac{N_k(s-1)}{\ln 2}\right]+ \pFA.
		\label{eq:mellin_steady}
	\end{aligned}
	\end{equation}
\end{cor}

\begin{proof}

It follows by taking the expectation of $h(\gamma_k,X_k)^{s-1}$, observing that $X_k$ independent of $\gamma_k$ and Bernoulli distributed with $p_X=\mathbb{P}(X_k=1)$, and that $p_X = 1-\pFA$ in the baseline scenario. For mathematical details, we refer to \cite{Forssell2017,Schiessl2015}.

\end{proof}

\subsection{Analysis for Sybil Attacks}

In a Sybil attack, the number of resources each device gets assigned $N_k \sim p_N(n)$ is varying depending on the success of the adversary. We provide the Mellin transform of the service process in this generalized case in the following corollary (following from Theorem~\ref{thm:mellin_simo}):

\begin{cor}
	Under Sybil attack, with scheduled resources distributed according to $p_N(n)$ and frame-drops with the false alarm rate $\pFA$, the service-process Mellin-transform is given by
	
	\equationspace
	\begin{align*}
		&{\MelS}_{,\text{Sybil}}(s) = \\ &(1-\pFA) \sum_{n}\left[\mathcal{M}_{g(\gamma_k)}\left[1+\frac{n(s-1)}{\ln 2}\right]p_N(n)\right]+ \pFA.
		\label{eq:mellinS_sybil}
		\numberthis
	\end{align*}
\end{cor}

\begin{proof}
Here, $h(\gamma_k,X_k)$ in \eqref{eq:h_function} is a function of the random variable $N_k$ as well. Following the same logic as in the proof of Corollary~\ref{cor:mellin_steady}, we find that

\equationspace
\begin{align*}
&\mathcal{M}_{h(\gamma_k,X_k,N_k)}(s)= \mathbb{E}_{\gamma_k,X_k,N_k}\left[h(\gamma_k,X_k,N_k)^{s-1}\right] \\
 & = p_X \sum_{n}\left[\mathbb{E}_{\gamma_k}\left[h(\gamma_k,1,n)^{s-1}\right]p_N(n)\right] + (1-p_X).
 \label{eq:mellinh_sybil}
\numberthis
\end{align*}
Similarly to Corollary~\ref{cor:mellin_steady}, we have

\equationspace
\begin{equation}
	\mathbb{E}_{\gamma_k}\left[h(\gamma_k,1,n)^{s-1}\right] = \mathcal{M}_{g(\gamma_k)}\left[1+\frac{n(s-1)}{\ln 2}\right],
\end{equation}
and by plugging this into \eqref{eq:mellinh_sybil} and again noting that $p_X = 1-\pFA$, the proof of \eqref{eq:mellinS_sybil} follows.
\end{proof}

\subsection{Analysis for Disassociation Attacks}
\label{subsec:mod_dis}

The modifications to the queueing model for disassociation attacks are described in Section~\ref{sec:disassociation_attack}. The delay bound in Section~\ref{sec:snc} applies to the block-aggregated service and arrival processes $s'_l$ and $a'_l$. However, with the redefinition of the time scale, we now have $p(w)=\Pr(W(t)>\KRC w)$. In the following, we present the Mellin transforms of the aggregated arrival and service process under these assumptions. Since we assume constant arrivals of $\alpha$ bits per frame, we simply have $a'_l=\KCN\alpha$ and ${\MelA}_{\text{,Disassociation}}(s) = e^{\KCN\alpha(s-1)}$. What remains is the Mellin transform of the aggregated service process, provided in the following corollary:

\begin{cor}
	For the $\KCN$ aggregated service process under a disassociation attack with success probability $p_d$, the Mellin transform is given by
	\equationspace	
	\begin{equation}
		{\MelS}_{\text{,Disassociation}}(s) = (1-p_d)\left[{\MelS}_{,\text{Baseline}}(s)\right]^{\KCN}+p_d,
	\end{equation}
	where ${\MelS}_{,\text{Baseline}}(s)$ is the Mellin transform for the baseline scenario given by \eqref{eq:mellin_steady}.
\end{cor}

\begin{proof}
We note that ${\MelS}_{\text{,Disassociation}}(s)$ is given by

\equationspace
\begin{align*}
	 \E[e^{s'_l(s-1)}] & =(1-p_d)\E\left[\left(\prod_{k=\KCN l}^{\KCN (l+1)-1}h(\gamma_k,X_k)\right)^{s-1}\right] + p_d \\
	&= (1-p_d)\left[{\MelS}_{,\text{Steady}}(s)\right]^{\KCN} + p_d,
\end{align*}
where we have used that $h(\gamma_k,X_k)$ is independent for each $k$.

\end{proof}

\vspace{-3ex}
\section{Numerical Results}
\label{sec:results}
\newcommand\resultFigWidth{.325\textwidth}
\newcommand\subFigWidth{.330\textwidth}

In this section, we use our analytical results to study a network consisting of $K_d=24$ MTC devices deployed in a square 20 m $\times$ 20 m grid, one access point placed at the origin, and the adversary Eve positioned either outside the network or representing a compromised device within the network (see Fig.~\ref{fig:simulation_deployment} for an example deployment). The network is operating at carrier frequency $f_c = 2.4$ GHz and the access point antenna array has normalized antenna separation $\Delta_r=0.5$ and is oriented parallel to the line $y=-x$. We assume that all channel covariance-matrices are of the form $[\A]_{i,j} = \rho^{|i-j|}$ where $\rho$ is a correlation coefficient. The delay performance bounds are upper bounds on the delay violation probability $p_i(w) < p_{i,\text{Bound}}(w)$ computed as described in Section~\ref{sec:snc}, where the minimization over $s$ is carried out by a grid search. To specify the arrival rates $\alpha_i$, we compute the rate corresponding to a fixed server utilization, defined as $u=\frac{\E(a_k)}{\E(s_k)}$.

\begin{figure*}
\centering
\begin{subfigure}{\subFigWidth}
\psfrag{X}[][]{\scriptsize $x$ [m]} %
\psfrag{Y}[][]{\scriptsize $y$ [m]} %
\includegraphics[width=\columnwidth]{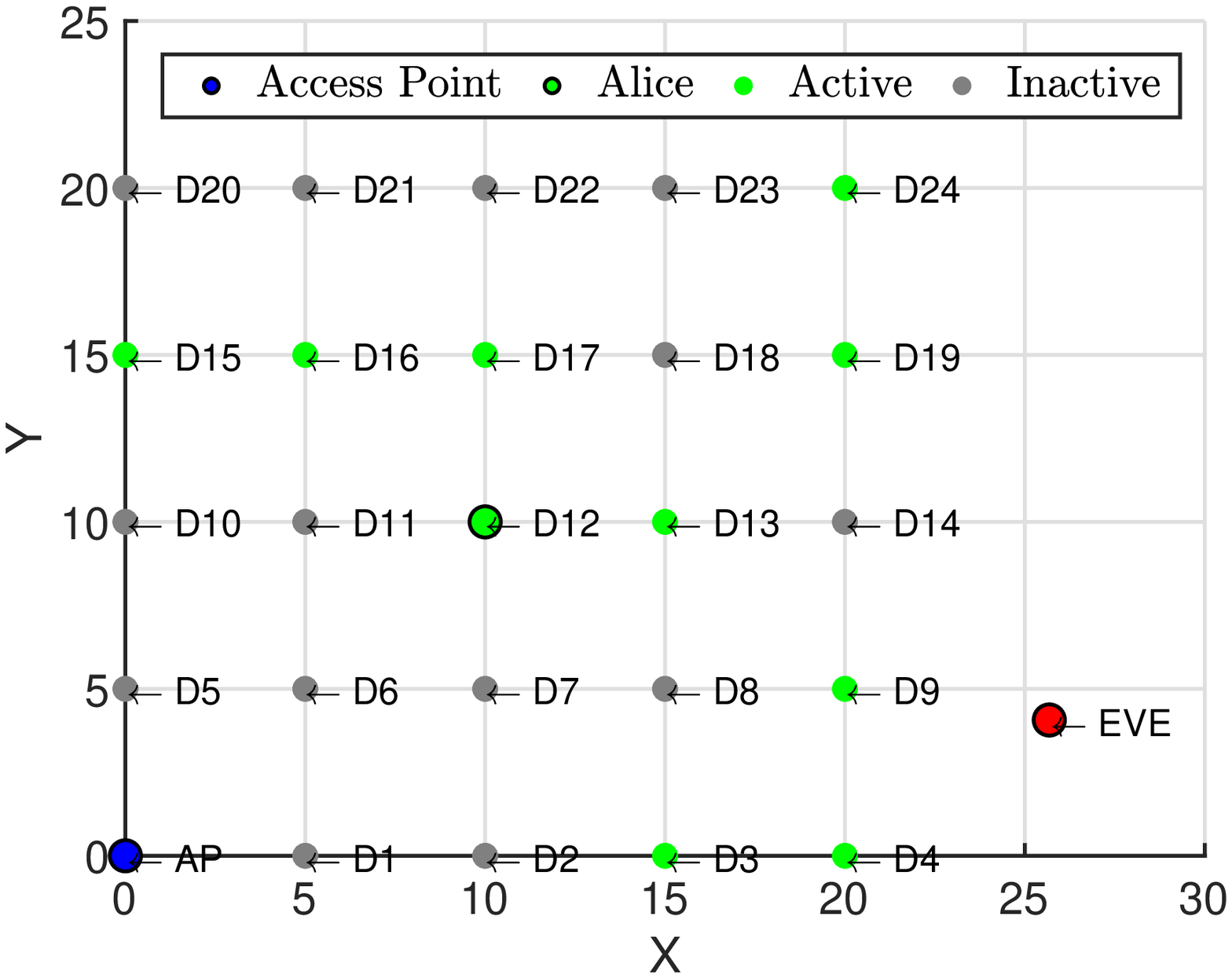}
\caption{}
\label{fig:simulation_deployment}
\end{subfigure}
\begin{subfigure}{\subFigWidth}
\psfrag{X}[][]{\scriptsize Delay [frames]} %
\psfrag{Y}[][]{\scriptsize Delay Violation Probability} %
\includegraphics[width=\columnwidth]{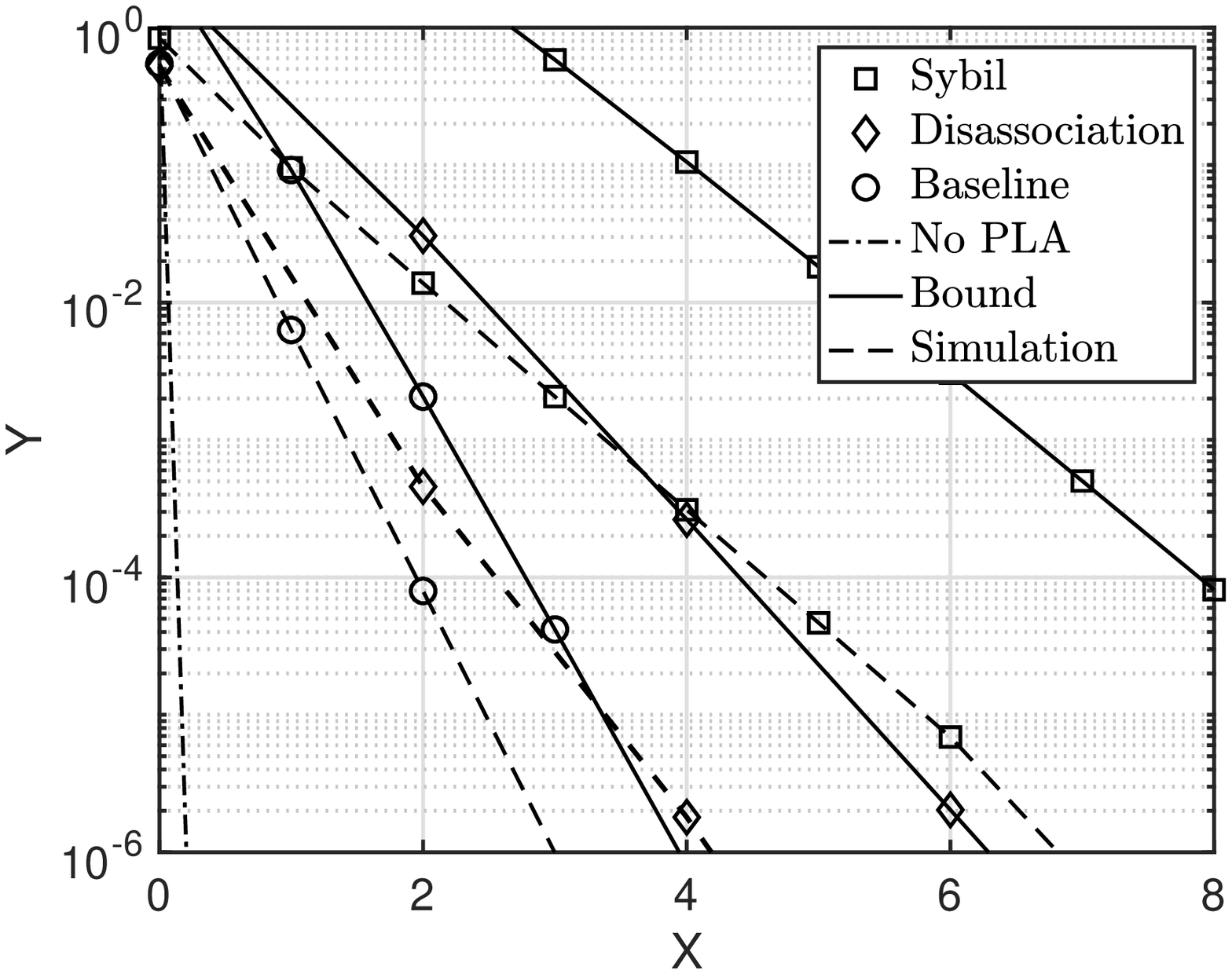}
\caption{}
\label{fig:r_bound_validation}
\end{subfigure} 
\begin{subfigure}{\subFigWidth}
\psfrag{X}[][]{\scriptsize False Alarm Rate} %
\psfrag{Y}[][]{\scriptsize $w_\epsilon$ [frames]} %
\psfrag{A}{\tiny \colorbox{white!30}{$u=0.9$}} %
\psfrag{B}{\tiny \colorbox{white!30}{$u=0.5$}} %
\includegraphics[width=\columnwidth]{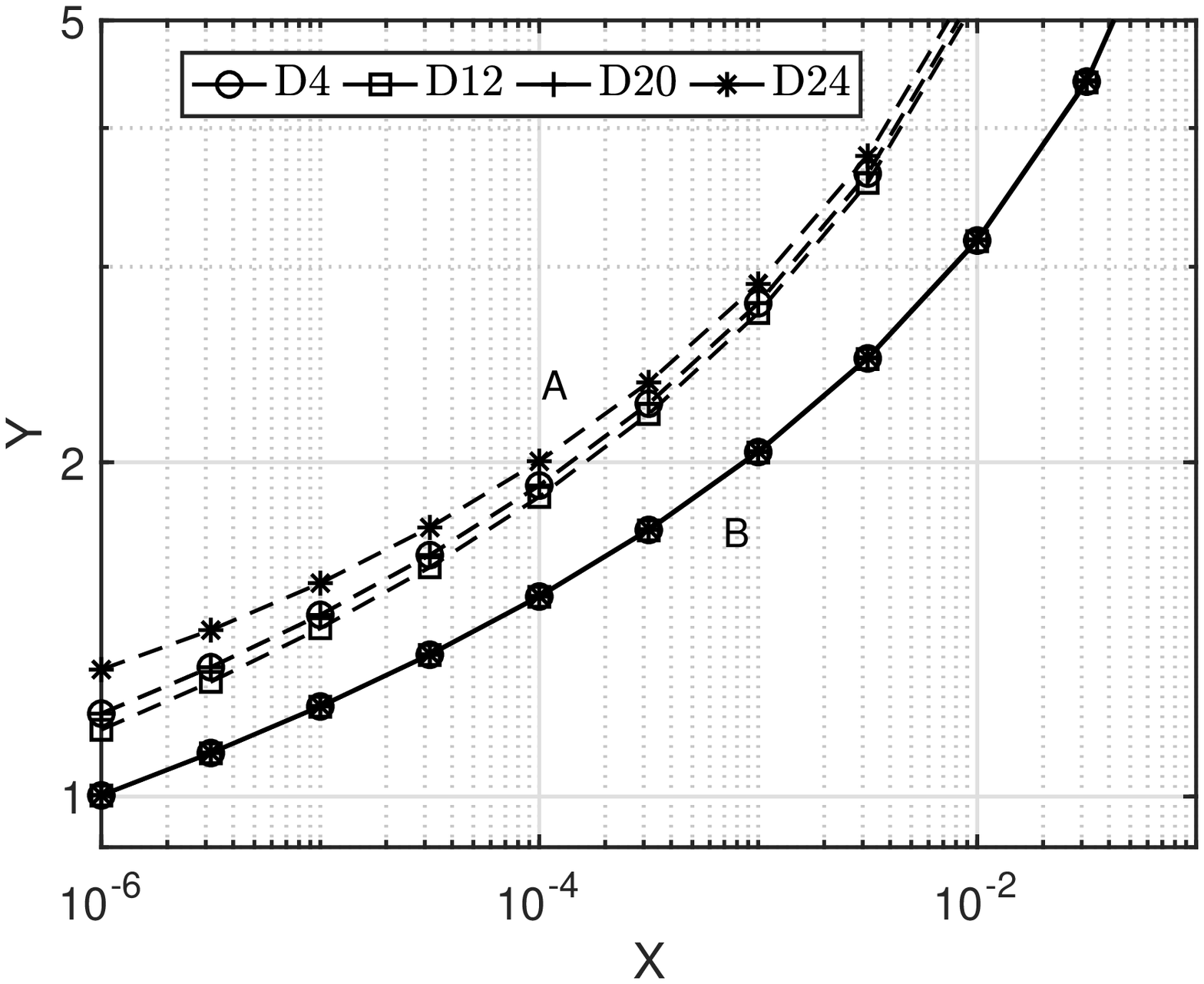}
\caption{}
\label{fig:r_steady_state_delay_vs_far}
\end{subfigure}  \vspace{-1ex}
\caption{(a) Considered MTC deployment. (b) Comparison of link-level simulations to derived bounds for device D12. (c) Delay guarantee $w_\epsilon$ with $\epsilon = 10^{-6}$ for device D12 in baseline scenario.}
\vspace{-2ex}
\end{figure*}

\vspace{-1ex}
\subsection{Bound Validation}
In Fig.~\ref{fig:r_bound_validation}, we show the delay violation probability for device D12, evaluated through link-level simulations, together with the corresponding bounds. In this figure, $\KRice = 6$ dB, $\rho=0$, $\NRx =4$ and we have included results with PLA in the baseline scenario, under Sybil attack, and under disassociation attack, as well as results without PLA in the baseline scenario. For the Sybil attack, we assume Eve launches $|\Dsybil|=4$ Sybil IDs, and for the disassociation attack we assume the reconnection time is $\KCN = 4$ frames. We can observe that in each scenario the analytical bounds and simulation results follow the same slope with a gap of 1-3 orders of magnitude between the curves; hence, our analysis can validly upper bound the performance of the modeled system. Typically, the gap between simulation and bound increases with the slope of the curves. Therefore, the derived bounds will clearly overestimate the true delay violation probability as seen in Fig.~\ref{fig:r_bound_validation}. However, they provide us with an efficient (i.e., in the sense that computing the bounds is significantly less computationally demanding than simulating the system) and conservative (i.e., the true system will perform considerably better than the bounds predict) way of evaluating the system's delay performance.

\vspace{-0.5ex}
\subsection{Baseline Performance}

Fig.~\ref{fig:r_steady_state_delay_vs_far} illustrates the delay $w_{\epsilon}$ that can be analytically guaranteed with a violation probability of $\epsilon=10^{-6}$, i.e., $p_{i,\text{Bound}}(w_{\epsilon}) = \epsilon$, for a given false alarm rate. For illustration, we consider only a subset of devices covering the full range of AoAs and distances. The value of $w_\epsilon$ varies little between devices. This is beacuse arrival rates are adapted differently for each device to fix the utilization $u$. Fig.~\ref{fig:r_steady_state_delay_vs_far} also illustrates how PLA impacts the system: to get a low missed detection rate we typically want to have a low threshold $T$. However, decreasing $T$ increases the false alarm rate, which clearly impacts the delay performance guarantee $w_\epsilon$. For this particular scenario, we see that a false alarm rate approaching $10^{-2}$ can have an impact of 2-5 frames on the delay guarantee. Higher utilization in Fig.~\ref{fig:r_steady_state_delay_vs_far} means that the arrival rates are higher, resulting in an increased delay guarantee. However, we observe that the delay shows a similar behavior with the false alarm rate for both $u=0.5$ and $u=0.9$. 

\subsection{Data Injection Attacks}

Here, we consider the detection performance of PLA in the data injection attack. Delay impacts of data injection are not studied in this case since these are similar to the Sybil attack, as discussed in Section~\ref{sec:data_injection}.
\paragraph{Inactive device targeted}

Fig.~\ref{fig:r_MDR} shows the analytical missed detection rate $\pMD(i,T)$ when Eve launches a data injection attack against an inactive device. In these figures, we assume that Eve is positioned at $(x,y) = (25,0)$ [m], targeting devices $\{\text{D}4,\text{D}8,\text{D}12,\text{D}16,\text{D}20\}$, and that the PLA threshold is fixed at a false alarm rate $\pFA=10^{-2}$. Fig.~\ref{fig:r_mdr_vs_RiceK} depicts the missed detection rate for varying $\KRice$ with fixed $\KRiceE=0$ dB. As expected, the missed detection rate improves with stronger LOS component. We observe that the detection performance for D4 is worse due to its location at $(0,20)$ close to Eve. Additionally, we can observe that higher antenna correlation has a positive effect on the missed detection rate performance. Fig.~\ref{fig:r_mdr_vs_RiceK_E} shows the influence of $\KRiceE$ on the missed detection rate for fixed $\KRice=5$ dB. For PLA of devices far from Eve, a stronger LOS component from Eve allows the access point to better differentiate messages from Eve. However, for device D4 the missed detection rate shows the opposite behavior since Eve's channel more and more resembles the legitimate channel. We can also see that for low $\KRiceE$ (i.e, Eve's channel is approaching NLOS), the missed detection rate approaches the same value for all choices of devices to impersonate. In Fig.~\ref{fig:r_mdr_vs_NRx}, we plot the missed detection rate for varying $\NRx$ showing that the missed detection rate follows an approximately log-linear decrease with $\NRx$.

\vspace{-1ex}
\begin{figure*}
\centering
\begin{subfigure}{\subFigWidth}
\psfrag{X}[][]{\scriptsize $\KRice$ [dB]} %
\psfrag{Y}[][]{\scriptsize Missed Detection Rate} %
\includegraphics[width=\columnwidth]{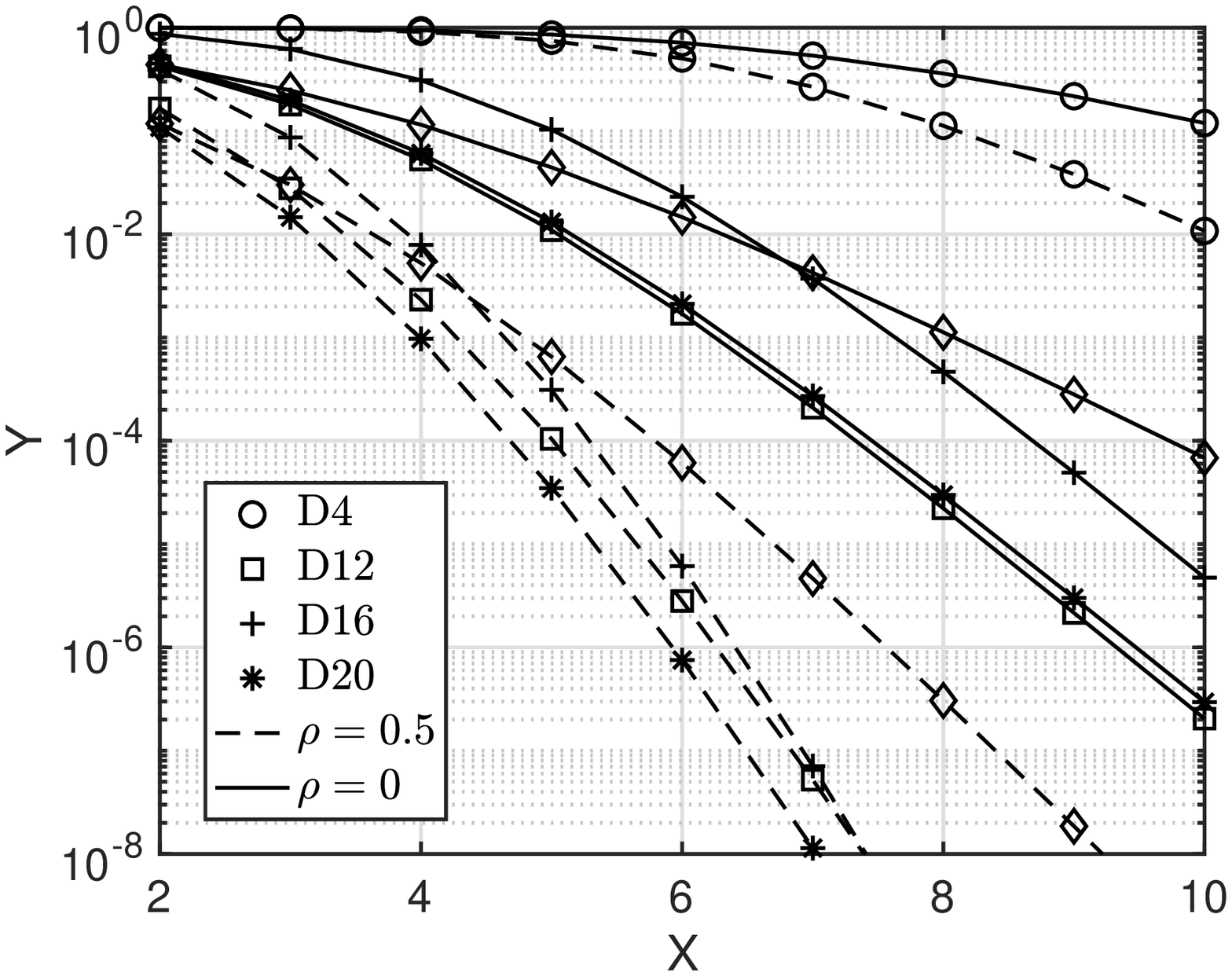}
\caption{}
\label{fig:r_mdr_vs_RiceK}
\end{subfigure}
\begin{subfigure}{\subFigWidth}
\psfrag{X}[][]{\scriptsize $\KRiceE$ [dB]} %
\psfrag{Y}[][]{\scriptsize Missed Detection Rate} %
\includegraphics[width=\columnwidth]{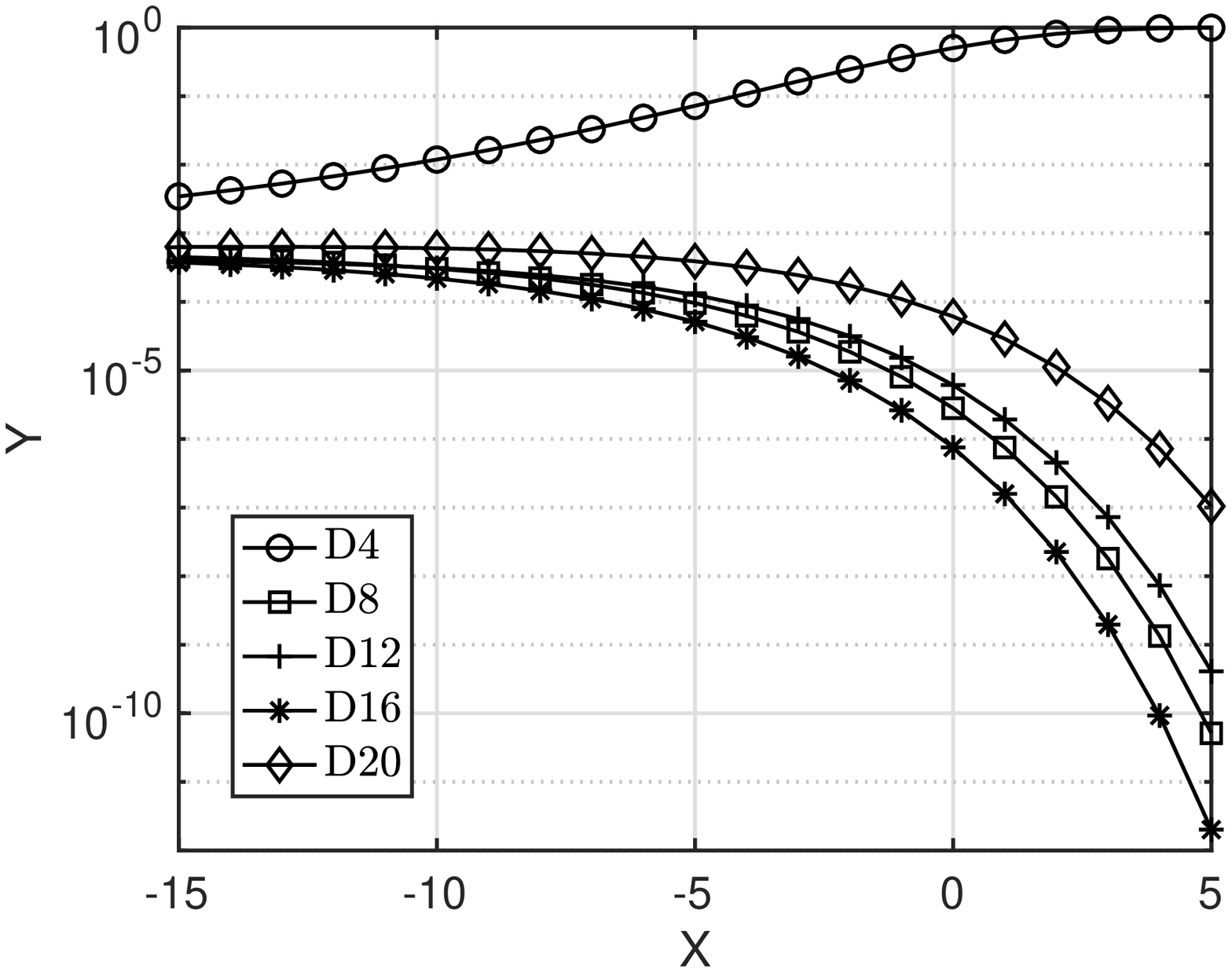}
\caption{}
\label{fig:r_mdr_vs_RiceK_E}
\end{subfigure} 
\begin{subfigure}{\subFigWidth}
\psfrag{X}[][]{\scriptsize $\NRx$} %
\psfrag{Y}[][]{\scriptsize Missed Detection Rate} %
\includegraphics[width=\columnwidth]{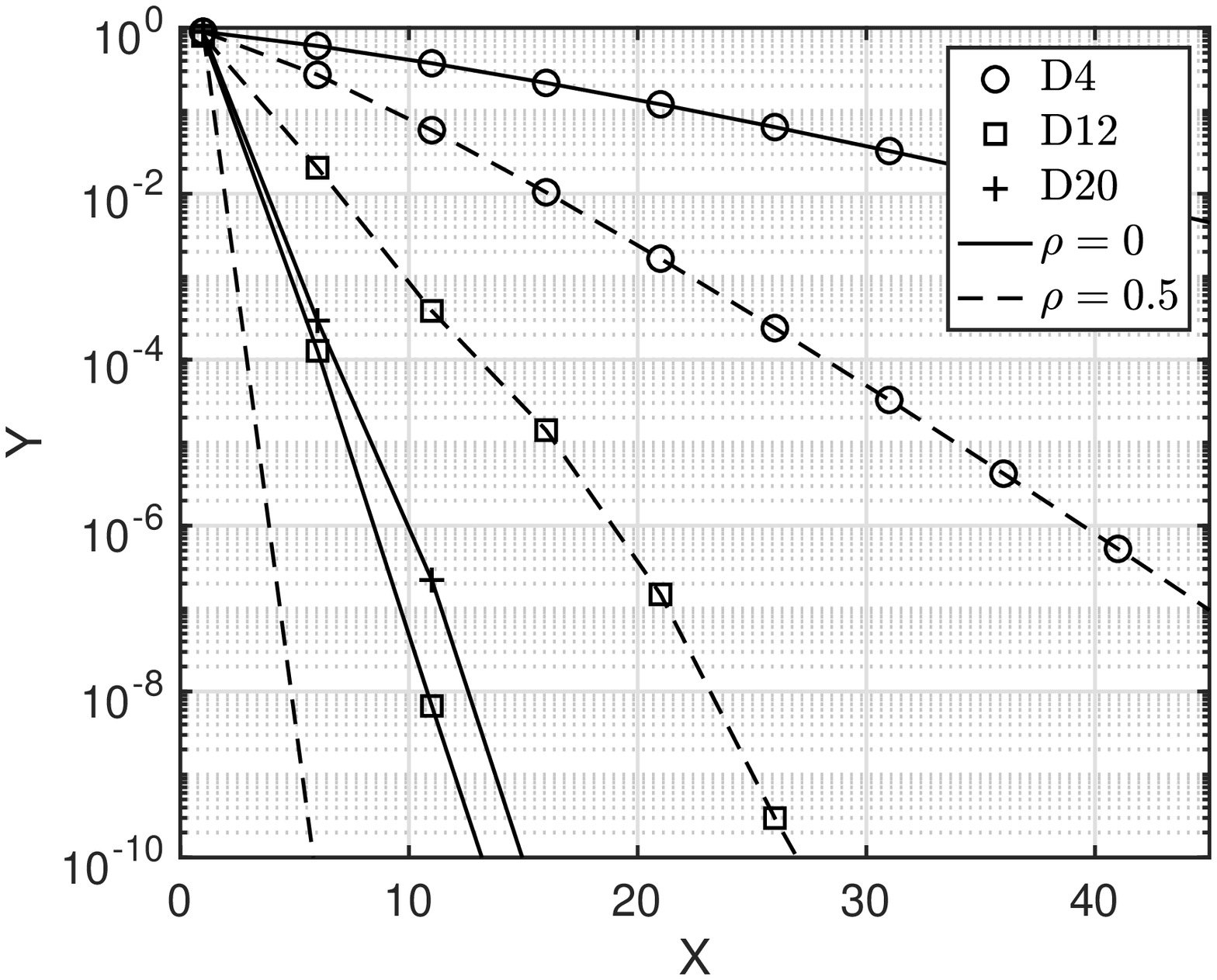}
\caption{}
\label{fig:r_mdr_vs_NRx}
\end{subfigure}  \vspace{-1ex}
\caption{PLA detection performance under data injection attack with $\NRx = 8$ when Eve impersonates $\{\text{D}4,\text{D}8,\text{D}12,\text{D}16,\text{D}20\}$: (a) For varying LOS strength, (b) for varying attacker LOS strength, and (c) for varying number of receive antennas.}
\vspace{-2ex}
\label{fig:r_MDR}
\end{figure*}

\paragraph{Active device targeted}

Fig.~\ref{fig:r_pe_vs_AoA} shows the missed detection rate $\pMD^{L=2}$ when Eve targets an active device from $d_E=30$ m and varies her AoA from $\pi/4$ to $3\pi/4$. The first upper bound corresponds to $\pMD(T,i)$, the second to $p_d(i)$, the lower bound correspond to $\pMD(T,i)\pFA(T)$ (see \eqref{eq:upper_lower_p_e} and \eqref{eq:p_e_upper} in Section~\ref{sec:data_injection}), and the solid curve is generated by Monte Carlo simulation. We can observe that the gap between the tightest upper bound and the true value is around 1 order of magnitude. Additionally, this figure illustrates that there is an optimal AoA for Eve to impersonate this particular device with the highest success rate. In Fig.~\ref{fig:r_pe_optimal_AoA_v2}, we depict the upper bound on $\pMD^{L=2}$, for each device, when Eve is choosing the optimal AoA. Note that this is the upper bound and that the actual detection performance is around one order of magnitude lower. We observe that for devices D1, D5 and D6, the missed detection rate is very low ($<10^{-8}$), while the upper bound can approach values higher than $10^{-1}$ for some poorly positioned devices. Also, we observe that generally, the missed detection rate is improved when Eve only has a NLOS channel (i.e., $\KRiceE = -\infty$ dB). As a single-antenna attacker, we note that it is impossible for Eve to estimate the optimal AoA through eavesdropping communications. Though through knowledge of the deployment, Eve can position herself at a similar LOS path as the legitimate device to optimize her chances of success. However, if Eve's objective is to impersonate several devices simultaneously, the optimal AoA becomes conflicting as illustrated by Fig.~\ref{fig:r_pe_optimal_AoA_v2}.

These results highlight two variables affecting the detection performance of PLA for a given false alarm rate: (i) network deployment and environment affecting LOS strengths for legitimate channels and for Eve; and (ii) access point design in terms of amount and placement of antennas. It is clear that we can improve the missed detection rate by adding more antennas and placing them such that antenna-correlation is high. Moreover, by designing the deployment and the immediate environment such that devices have a strong LOS path to the access point, while Eve is unable to get a strong LOS path (e.g., through deployment of the system in a closed environment), we can improve detection performance. Influencing channel characteristics for improved PLA performance might be feasible in some scenarios (e.g., in a factory deployment). Moreover, deployments with strong LOS components might be desirable for pure communication reasons as well.

\vspace{-1ex}
\begin{figure*}
\centering
\begin{subfigure}{\subFigWidth}
\psfrag{X}[][]{\scriptsize Attacker AoA [$\pi$ rad]} %
\psfrag{Y}[][]{\scriptsize $\pMD^{L=2}$} %
\includegraphics[width=\columnwidth]{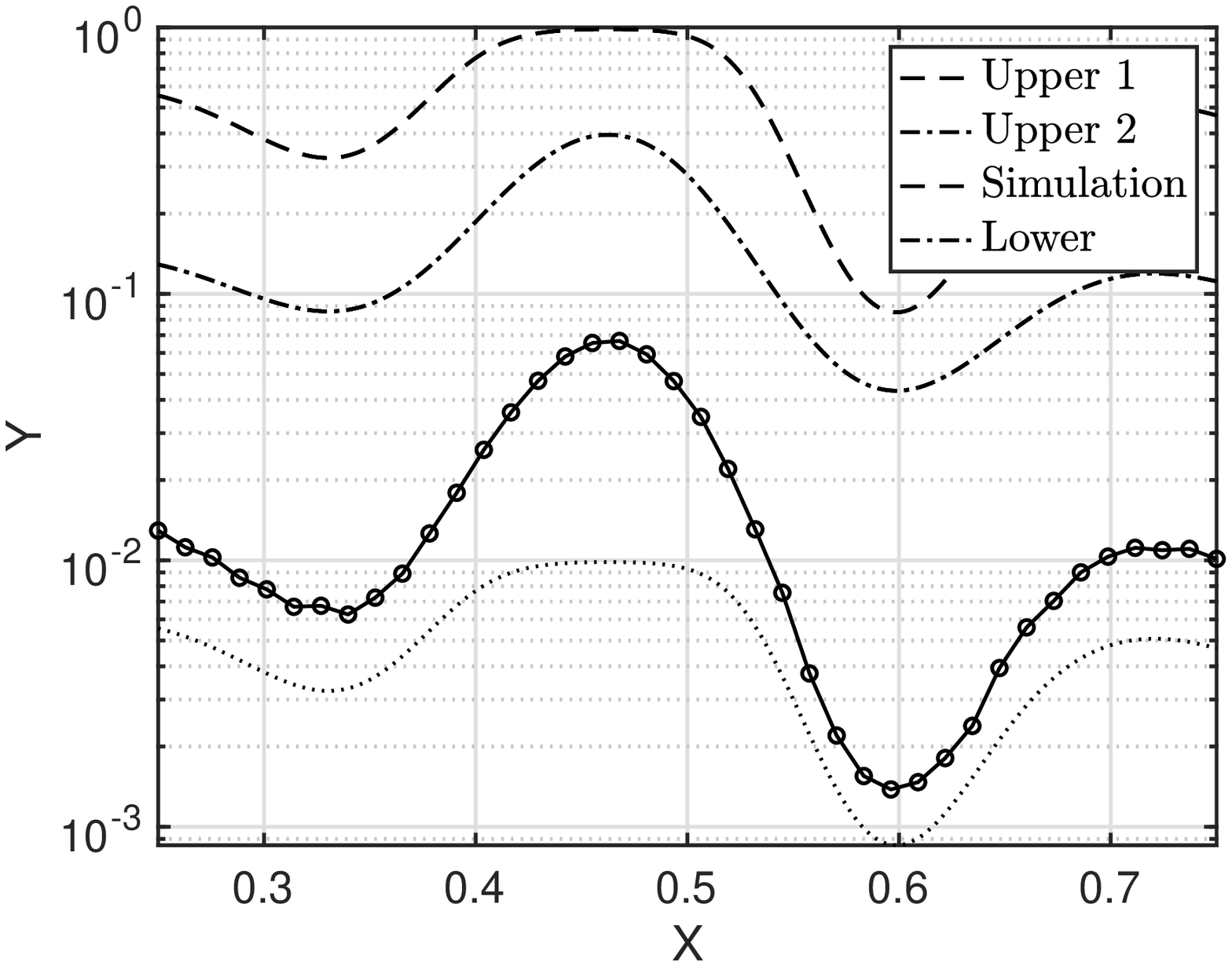}
\caption{}
\label{fig:r_pe_vs_AoA}
\end{subfigure}\hspace{10ex}
\begin{subfigure}{\subFigWidth}
\psfrag{X}[][]{\scriptsize Device Index} %
\psfrag{Y}[][]{\scriptsize $\pMD^{L=2}$} %
\psfrag{Z}[][]{\scriptsize Optimal AoA [$\pi$ rad]} %
\includegraphics[width=\columnwidth]{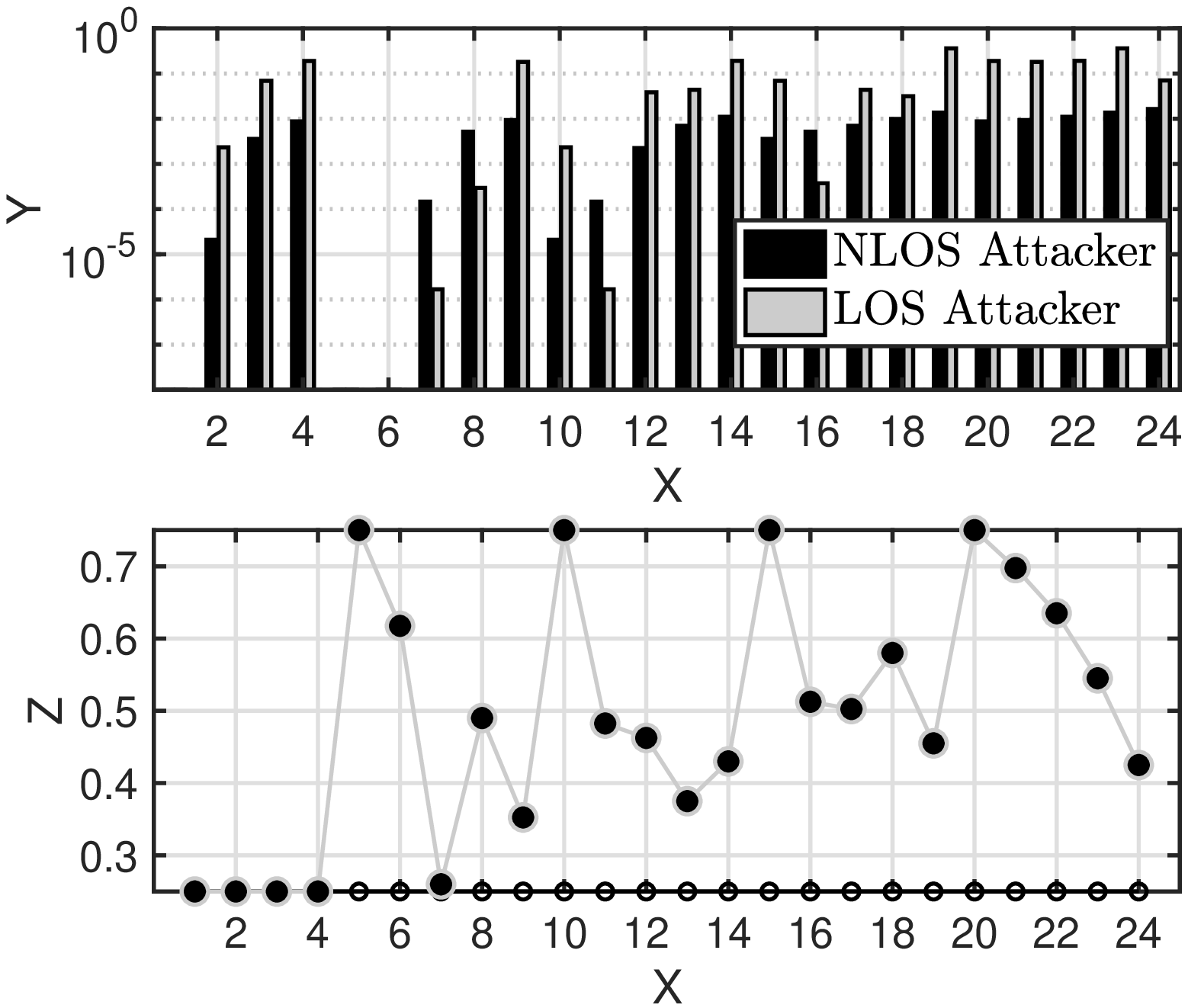}
\caption{}
\label{fig:r_pe_optimal_AoA_v2}
\end{subfigure}  \vspace{-1ex}
\caption{Data injection attack when Eve impersonates an active device: (a) missed detection rate vs. attacker AoA when D12 is targeted, $\KRice = 6$ dB and $\KRiceE=0$ dB, (b) Upper bound on missed detection rate and Eve's optimal AoA.}
\vspace{-2ex}
\end{figure*}

\vspace{-0.5ex}
\subsection{Sybil Attacks}

Here, we assume that devices $\Dactive =\{\text{D}12,\text{D}13,\text{D}14,\text{D}17,\text{D}18,\text{D}19,\text{D}22,\text{D}23,\text{D}24\}$ (i.e., the upper-right quadrant of the deployment) are active and that device D4 has been compromised and is launching a Sybil attack. Fig.~\ref{fig:r_sybil_average_sybil_nodes} shows $\E[\Ksybil]$, the average number of Sybil nodes successfully launched by Eve, as a function of the number of targeted devices $|\Dsybil|$. The solid lines are computed according to our approximate distribution \eqref{eq:k_sybil_approx_pmf}, while the dashed lines correspond to simulation results, showing that our approximation is accurate. We see that with no PLA, Eve successfully gets every Sybil ID accepted. With PLA and lower $\pFA$, the number of successful Sybil nodes is kept lower. For instance, when $\pFA = 10^{-2}$, the expected number of Sybil nodes does not exceed $\E[\Ksybil]>2$ even though Eve can launch up to $|\Dsybil|=14$ Sybil IDs, which means that PLA is effective against the attack. However, it is apparent from Fig.~\ref{fig:r_sybil_average_sybil_nodes} that there are Sybil IDs that cannot be detected by the PLA. The reason is that in the particular scenario that we have investigated, Eve is device D4, and hence, more easily impersonates devices $\{\text{D}1,\text{D}2,\text{D}3\}$ due to having the same AoA. Fig.~\ref{fig:r_sybil_delay_vs_sybil_nodes} shows the delay guarantee $w_\epsilon$ for D12 and $\epsilon=10^{-6}$ under the Sybil attack. We can observe that without PLA, the increasing number of Sybil IDs launched by Eve has a severe effect on the delay performance. For example, when the link utilization is high ($u=0.9$), Eve only has to introduce 4-5 Sybil IDs to cause the delay in the queue to grow towards infinity. On the other hand, by effectively detecting the Sybil IDs with PLA with $\pFA=10^{-2}$, the delay performance can be made almost independent of the number of Sybil IDs, at a cost of a constant higher delay of around 3 frames.

\vspace{-1ex}
\begin{figure*}
\centering
\begin{subfigure}{\subFigWidth}
\psfrag{A}[][]{\scriptsize $|\Dsybil|$} %
\psfrag{B}[][]{\scriptsize $\E[\Ksybil]$} 
\includegraphics[width=\columnwidth]{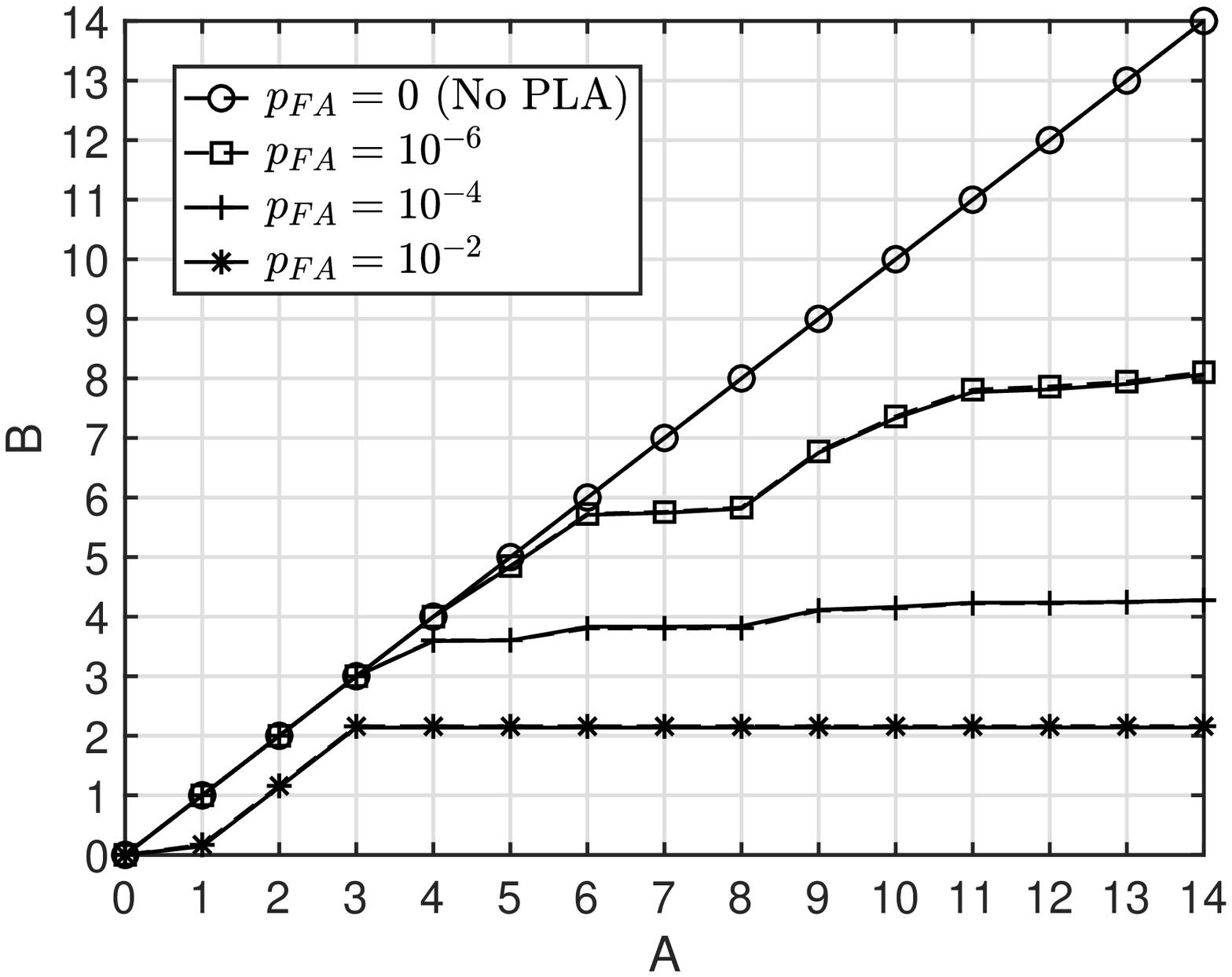}
\caption{}
\label{fig:r_sybil_average_sybil_nodes}
\end{subfigure}
\begin{subfigure}{\subFigWidth}
\psfrag{X}[][]{\scriptsize $|\Dsybil|$} %
\psfrag{Y}[][]{\scriptsize $w_\epsilon$ [frames]} %
\psfrag{B}[][]{\colorbox{white!30}{\tiny$\pFA=10^{-2}$}}
\psfrag{C}[][]{\colorbox{white!30}{\tiny No PLA}} %
\psfrag{D}{\tiny $w_\epsilon \rightarrow \infty$} %
\includegraphics[width=\columnwidth]{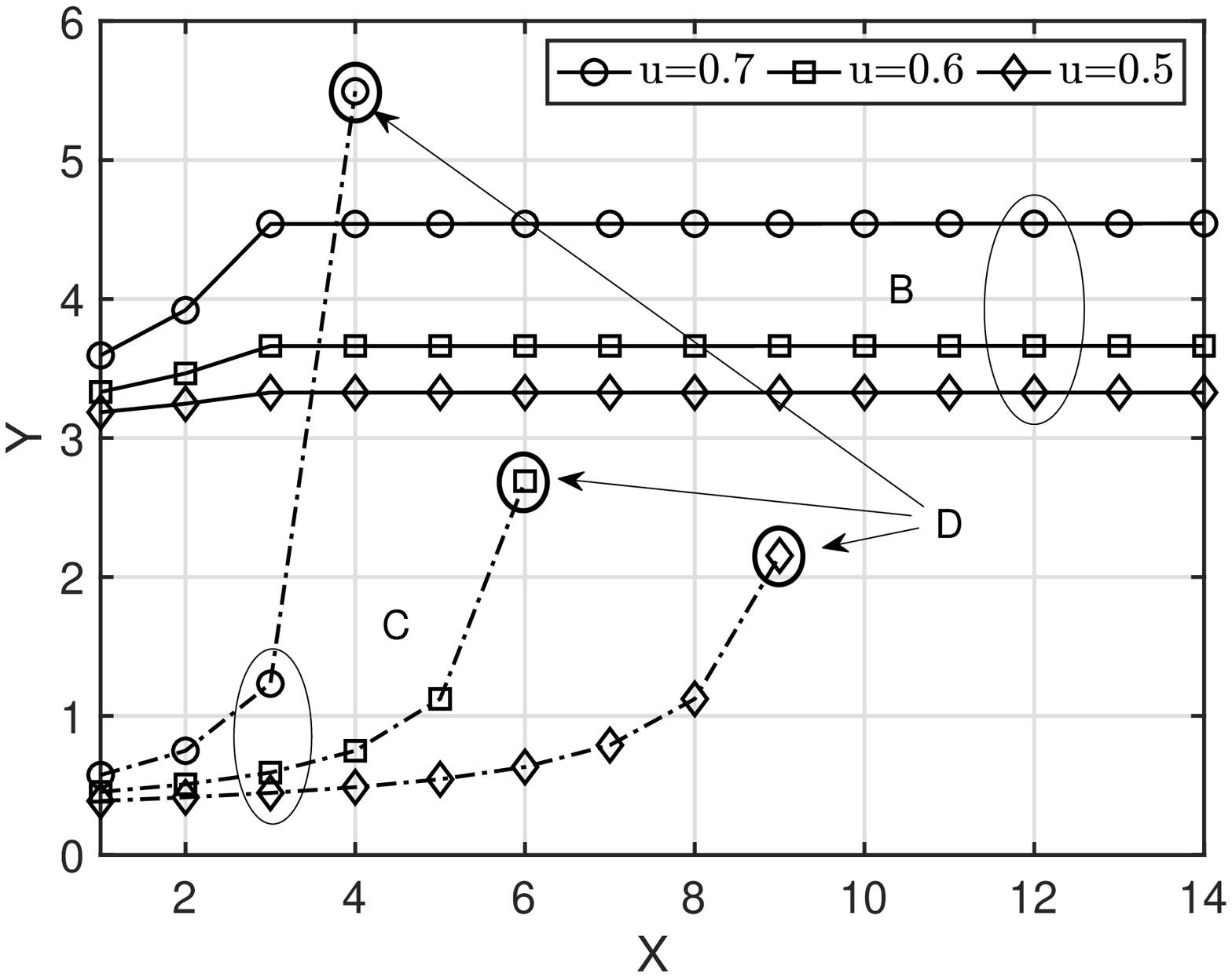}
\caption{}
\label{fig:r_sybil_delay_vs_sybil_nodes}
\end{subfigure} 
\begin{subfigure}{\subFigWidth}
\psfrag{X}[][]{\scriptsize $\pAttack$} %
\psfrag{Y}[][]{\scriptsize $w_\epsilon$ [frames]} %
\psfrag{A}[][]{\colorbox{white!30}{\tiny $\NRx=4$}} %
\psfrag{B}[][]{\colorbox{white!30}{\tiny $\NRx=8$}} %
\psfrag{C}{\colorbox{white!30}{\tiny No PLA}} %
\includegraphics[width=\columnwidth]{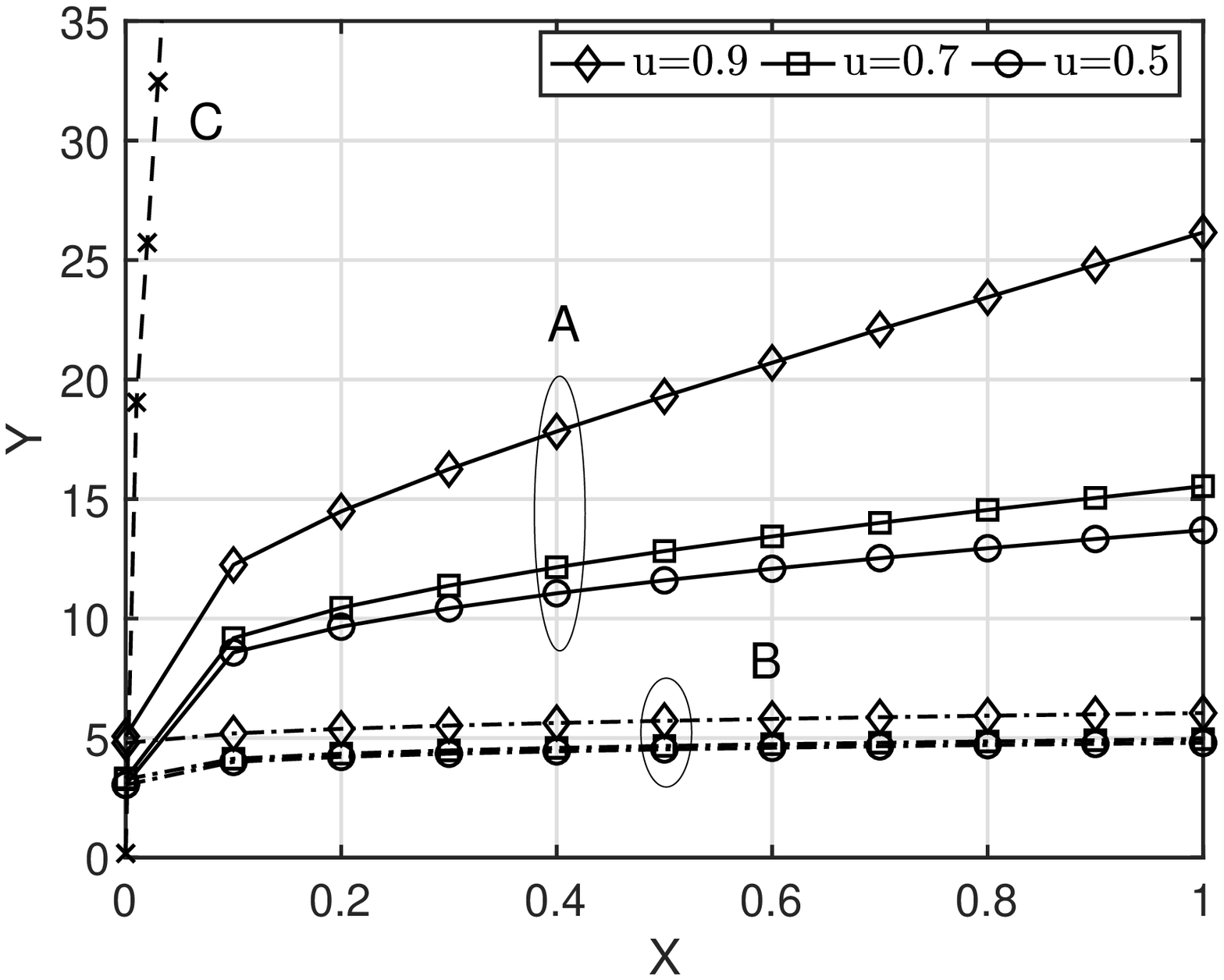}
\caption{}
\label{fig:r_dis_delay_vs_attack_prob}
\end{subfigure} \vspace{-1ex}
\caption{(a) Expected number of successful Sybil IDs $\E[\Ksybil]$ for various choices of $\pFA$. (b) Delay performance impacts for D12 under Sybil attack. (c) Delay performance impacts for D12 under disassociation attack.}
\vspace{-2ex}
\end{figure*}
\vspace{-1ex}
\subsection{Disassociation Attacks}

Here, we assume that Eve is an external entity, positioned at $d_E=25$ m and $\Phi_E = \pi/3$, launching a disassociation attack targeting device D12. Fig.~\ref{fig:r_dis_delay_vs_attack_prob} shows the delay guarantee $w_\epsilon$ for $\epsilon=10^{-6}$ for the targeted device as a function of the attack probability. In the results without PLA, we have assumed that the access point performs random guessing whenever two requests are received at the same time. For this case, we can clearly see that the attack causes the delay to increase for very low attack intensities, simply because the device gets disconnected 50\% of the times Eve sends a DCN request. With PLA the impact is reduced and the system is able to give delay guarantees even when $\pAttack \rightarrow 1$ and Eve is sending DCN requests in every frame. For $\NRx = 4$, we however see an increase in $w_\epsilon$ with $\pAttack$ due to the occasional missed detections. Fig.~\ref{fig:r_dis_delay_vs_attack_prob} also illustrates that this increase can be mitigated by increasing the number of receive antennas to $\NRx =8$.

\subsection{Discussion}

The PLA scheme studied in this paper can achieve missed detection rates of around $10^{-6}$ and even as low as $10^{-10}$ under certain channel conditions.
These values can certainly meet security requirements even in applications where message integrity is of critical importance. 
Our results also indicate that these security enhancements comes at a limited cost in terms of delay:
In the baseline scenario, given reasonable false alarm rates $\pFA<10^{-2}$, our results show that a delay $w_\epsilon<5$ frames can be guaranteed with a reliability of $\epsilon=10^{-6}$. Furthermore, PLA can assure reliable operation of the studied system, even under hostile scenarios like Sybil and disassociation attacks that are aimed at impacting the service on legitimate channels. 
All these observations make a strong case for that PLA can be a viable option for message integrity in mission-critical MTC.

Introducing multiple-antenna access points appear beneficial from both a PLA security and a delay perspective. 
Already with 4-8 receive antennas, we observe large benefits in terms of missed detection rate that continues to improve in a log-linear fashion. 
The benefits of introducing more antennas at the access point can be interpreted in two ways: (i) improved detection performance for a given false alarm rate, or (ii), improved false alarm rate for a given detection performance. 
That is, we can utilize extra antennas either to strengthen the integrity of communication, or to reduce the delay impacts (i.e., decrese false alarm rate).
Our results also provide insight into deployment strategies: We have seen that if many devices are deployed along a straight line resulting in similar AoA profiles, an external or internal attacker will be more effective in impersonating this set of devices simultaneously. Hence, if the deployment of MTC devices can be influenced for security purposes, this can be used to make sure that Sybil attacks targeting many devices are unlikely to succeed. Furthermore, if certain devices transmit particularly sensitive information, these can be placed in positions such that Eve's success-rate when impersonating is minimized. 

\section{Conclusions}
\label{sec:conclusion}

We have studied delay impacts of a feature-based PLA protocol in order to investigate the viability of PLA for mission-critical MTC applications. 
Based on a MTC network model consisting of multiple devices and a multi-antenna access point we have derived delay performance bounds that quantify the delay impacts of PLA. 
Evaluation of the derived bounds for a network with a square-grid deployment of 24 MTC devices shows that PLA can, under good LOS conditions, be used without introducing excessive delays.
Additionally, we have found that PLA allows low-latency high-reliability communication even under hostile attack scenarios such as Sybil and disassociation attacks.
As a means of improving detection and delay performance, one could consider multiple antenna-arrays deployed at separate locations in a distributed manner. 
Additionally, in this paper we have limited the analysis to a single-antenna adversary; however, this could be extended to several adversaries with multiple antennas. 
Moreover, channel estimation techniques and their effect on the queueing model and authentication performance is still an open problem.
Finally, our analysis could easily be modified to encompass other authentication schemes (e.g., based on other features or fingerprinting tags) and through this be used to compare different PLA schemes from a delay perspective.

\vspace{-1ex}
\footnotesize
\bibliography{my_bibliography.bib}

\end{document}